\documentclass[a4p,11pt]{article}

\usepackage{amsfonts,amsmath,amssymb,amsthm,ascmac,mathrsfs,mathtools,comment,enumerate,bm,multicol,setspace,lmodern,microtype,physics,paralist,graphicx,hyperref,blindtext}
\usepackage[dvipsnames]{xcolor}
\usepackage[right=1.3in, left=1.3in, bottom=1.3in, top=1.3in]{geometry}
\hypersetup{colorlinks=true, linkcolor=black, citecolor=black}

\usepackage{natbib}
\setcitestyle{authoryear,open={(},close={)}}

\newtheorem{proposition}{Proposition}
\newtheorem{corollary}{Corollary}
\newtheorem{lemma}{Lemma}

\theoremstyle{definition}
\newtheorem{definition}{Definition}

\newtheorem*{assumption*}{Assumption}
\newtheorem{example}{Example}

\theoremstyle{remark}
\newtheorem{remark}{Remark}

\newcommand{\N}{\mathbb{N}}
\newcommand{\R}{\mathbb{R}}

\newcommand{\1}{\mathbf{1}}

\renewcommand{\d}{{\rm d}}
\newcommand{\s}{\mathfrak{s}}
\newcommand{\n}{\mathfrak{n}}

\newcommand{\calB}{\mathcal{B}}

\newcommand{\calL}{\mathcal{L}}

\newcommand{\calN}{\mathcal{N}}

\newcommand{\calX}{\mathcal{X}}

\newcommand{\esup}{\mathop{\rm ess~sup}\limits}

\DeclareMathOperator{\E}{\mathbb{E}}
\DeclareMathOperator{\Var}{\mathbb{V}ar}
\DeclareMathOperator{\Sd}{\mathbb{S}d}
\DeclareMathOperator{\Cov}{\mathbb{C}ov}
\DeclareMathOperator{\Corr}{\mathbb{C}orr}
\renewcommand{\P}{\mathbb{P}}

\DeclareMathOperator{\Span}{span}

\DeclareMathOperator{\co}{co}

\newcommand{\ton}{\xrightarrow{\|\cdot\|}}
\newcommand{\tow}{\xrightarrow{\rm w}}

\begin{document}

\title{{\bf On the Pettis Integral Approach to Large Population Games\thanks{We thank seminar participants at Tokyo University of Science, East Asian Contract Theory Conference, and SUSTECH-SZU Micro Theory Workshop for valuable comments and discussion.}}}
\author{Masaki Miyashita\footnote{The University of Hong Kong, \tt{masaki11@hku.hk}} \and Takashi Ui\footnote{Kanagawa University and Hitotsubashi University; \tt{takashi.ui@r.hit-u.ac.jp}}}
\date{\today}
\maketitle

\begin{abstract}
The analysis of large population economies with incomplete information often entails the integration of a continuum of random variables.
We showcase the usefulness of the integral notion \`a la \cite{pettis1938} to study such models.
We present several results on Pettis integrals, including convenient sufficient conditions for Pettis integrability and Fubini-like exchangeability formulae, illustrated through a running example.
Building on these foundations, we conduct a unified analysis of Bayesian games with arbitrarily many heterogeneous agents.
We provide a sufficient condition on payoff structures, under which the equilibrium uniqueness is guaranteed across all signal structures.
Our condition is parsimonious, as it turns out necessary when strategic interactions are undirected.
We further identify the moment restrictions, imposed on the equilibrium action-state joint distribution, which have crucial implications for information designer's problem of persuading a population of strategically interacting agents.
To attain these results, we introduce and develop novel mathematical tools, built on the theory of integral kernels and reproducing kernel Hilbert spaces in functional analysis.
\end{abstract}
\thispagestyle{empty}

\newpage
%\tableofcontents
\thispagestyle{empty}

%%%%%%%%%%%%%%%%%%%%%%%%%%%%%%%%%%%%%%%%%%%%%
%%%%%%%%%%%%%%%%%%%%%%%%%%%%%%%%%%%%%%%%%%%%%

\newpage
\hypersetup{linkcolor=Red, citecolor=Blue}
\setcounter{page}{1}

\section{Introduction}

\subsection{Aggregating Random Variables}

When analyzing economic models involving uncertainties associated with a large number of economic entities, we encounter the demand for aggregating a continuum of random variables.
For instance, in the beauty contest model introduced by \cite{morris2002shin}, each player best responds to an aggregated economic variable---an ``integral'' of actions taken by a continuum of opponents---by choosing an action equal to a convex combination of her best estimates regarding the aggregated action and an exogenous payoff-relevant state.
As opponents' actions are stochastic under incomplete information, the aggregated action takes the form of the integral of a ``stochastic process,'' which are indexed by the players' identities.
To fix an idea, let us consider the following symmetric linear-quadratic-Gaussian (LQG) game of \cite{bm2013}.\footnote{The LQG framework dates back to a seminal work by \citep{radner1962}.} %which is used as a running example that illustrates the relevance of our analysis to economic contexts.

\begin{example} \label{ex_bm}
Each agent $i$  in a closed interval $[0,1]$ chooses an action $a_i \in \R$ as a function of a private signal $x_i$ and a public signal $y$, which are correlated with a payoff-relevant random variable $\theta$, called the state.
We assume that these random variables are normally distributed as follows:
\begin{align*}
\theta \sim \calN \qty(\mu_\theta,\, \sigma^2_\theta),
\quad \epsilon_i \sim \calN \qty(0,\, \sigma^2_x),
\quad \epsilon \sim \calN \qty(0,\, \sigma^2_y),
\quad x_i = \theta + \epsilon_i,
\quad y = \theta + \epsilon,
\end{align*}
where $\epsilon_i$ and $\epsilon$ are independently distributed with respect to each other and to the state $\theta$.
Agent's payoff is determined by her own action $a_i$, the aggregated action $A \coloneqq \int_0^1 a_j \d j$, and the state $\theta$.
Specifically, we posit that in an equilibrium, agent $i$ sets her action $a_i$ equal to a linear combination of her best estimates regarding $A$ and $\theta$,
\begin{align} \label{br_bm}
a_i = r \E \qty[A \mid x_i,\, y] + s \E \qty[\theta \mid x_i,\, y] + k,
\end{align}
where $r,s,k \in \R$ are the parameters of the best response function.
The beauty contest is a special case that sets $r \in (0,1)$, $s = 1-r$, and $k = 0$.
\end{example}

In this example, the aggregated action $A$, given as the integration of a strategy profile, comes with some mathematical complications: As pointed out by \cite{judd1985}, a typical sample path of an i.i.d.\ process is terribly discontinuous, hampering the modeler to define the integral in a realized path-wise manner, which is one of the most natural approaches.
Since each individual action $a_j$ contains idiosyncratic randomness, coming from independent noise in private signals, defining $A$ faces this measurability issue.
Moreover, while one may expect a certain law of large numbers (LLN) to apply when aggregating a continuum of i.i.d.\ random variables, path-wise integration may lack this property: If we assume the idiosyncratic components disappear in the aggregate for all subintervals of a population, the stochastic process must essentially be constant, limiting the validity of the LLN to only trivial cases.\footnote{See Theorem 1 in \cite{uhlig1996}, as well as Proposition 2.1 in \cite{sun2006}.}

The source of these problems can be attributed to a way we interpret the integral of a stochastic process.
In this paper, instead, we advocate interpreting aggregation differently by appealing to the integral notion \`a la \cite{pettis1938} and demonstrates its usefulness in the analysis of large population games.\footnote{Some authors including \cite{al1995} and \cite{uhlig1996} also consider Pettis integral as a compelling remedy for the measurability problem, while they are mainly interested in applications to portfolio theory.}
The Pettis integral is defined for an abstract process that assigns to each input in a measurable space an output in a normed space.
When the process is valued in a Hilbert space, the core idea of Pettis can be described as: ``the inner product of the integral coincides with the integral of inner products.''
In a probabilistic context, this property implies: ``the expectation of the integral coincides with the integral of expectations,'' which bears similarity to Fubini's exchangeability.
In this regard, the Pettis integral is designed to retain the desirable property of usual integral notions, holding for a continuum of random variables.

Unlike other integral notions, Pettis integrability requires only a notably weak measurability condition, which can be met even when processes contain i.i.d.\ components.
In Proposition~\ref{prop_pettis}, we provide useful sufficient conditions of Pettis integrability, which can be stated solely in terms of the first and second moments of a stochastic process in the probabilistic context.
Moreover, we demonstrate the analytical tractability of Pettis integral by offering an extension of a Pettis integral version of the law of large numbers in \cite{uhlig1996} and a Fubini-like exchangeability formula that justifies the interchange of conditional expectation and integration.
In economic contexts, these results can aid, for example, in formalizing the derivation of a Bayesian Nash equilibrium or the moment restrictions of Bayes correlated equilibria in \cite{bm2013}.

%Moreover, unlike the standard setting of network games, our model entails incomplete information. Beyond extending the applicability of analysis, all these aspects generate novel theoretical implications regarding the equilibrium in network games. Since we do not a priori impose any substantial restrictions on the sign and magnitude of strategic interactions, the class of games is broad and can encapsulate many economically interesting settings that entail quadratic payoff functions.

\subsection{Large Population Games with Incomplete Information}

Leveraging the notion of Pettis integral, we offer a solid mathematical foundation for the equilibrium analysis of large population games with incomplete information.\footnote{While we mainly focus on large population games, our findings in Section~\ref{sec_eq} cover finite population games since the set of players is assumed as any finite measure space. In the case of finite populations, the Pettis integral of strategies simplifies to a linear combination of a finite number of random variables, avoiding any issues related to measurability.}
Our analysis is build on the hypothesis that each player's best-response strategy is linear in her best estimates regarding an \emph{idiosyncratic state} and a \emph{local aggregate}.
The former refers to an exogenous payoff-relevant random variable that may vary across agents, while the latter is defined as the weighted Pettis integral of other agents' strategies.
These weights can be heterogeneous and are specified by an exogenously given function, named a \emph{payoff structure}, which quantitatively delineates the connections between individuals within an economic or social network.
Our model bears similarity to a common setup in the network game literature but more general than it since we can accommodate arbitrarily many agents and uncertainty about the state.\footnote{For example, \cite{ballester2006} presents a benchmark model of network games with continuous actions and quadratic payoffs; see also Section 3.4 of the survey by \cite{jackson2015zenou}.
\cite{ozdaglar2023} extends the model to a continuum population by introducing a class of graphon games, modeling strategic interactions over a large scale network. Our payoff structure bears similarity to a graphon, a modeling device in \cite{ozdaglar2023}, but we do not restrict it to be non-negative, bounded, or undirected. Additionally, and more importantly, our model differs from theirs in that graphon games are complete information games without uncertain payoff states.}
In particular, incomplete information is a pivotal ingredient of the model, which yields novel theoretical implications regarding equilibrium properties, as well as enabling us to address information design problems.

In Proposition~\ref{prop_unique}, we address the issue of equilibrium uniqueness, a cornerstone in game theoretic analysis.
We establish conditions, dictating the uniqueness, in terms of the spectral properties of a payoff structure, which can be mathematically identified as an integral operator acting on real-valued functions.
We find that if the \emph{numerical range} of a payoff structure is bounded from above by $1$, equilibrium uniqueness prevails across ``all'' information environments, up to certain technical conditions pertaining to measurability and integrability.
Notably, the result does not rely on the normality of the state and signals, as opposed to the common assumption in LQG games.
Roughly speaking, the numerical range can be interpreted as measuring the magnitude and direction of strategic interactions.
Accordingly, the upper bounded condition can be interpreted as eliminating large-scale strategic complementarities, while the scale of strategic substitutes is not restricted as there is no lower bound condition.

Furthermore, a necessary condition is identified in the proposition:
If a payoff structure possesses an eigenvalue that is weakly greater than $1$, then there exists a distribution of agents' signals under which a continuum of equilibria emerge, provided that at least one equilibrium exists.
In other words, the violation of our necessary condition leads to either the non-existence or multiplicity of equilibria.
Notably, in the case of an undirected payoff structure, its maximal eigenvalue coincides with the supremum of its numerical range, rendering our sufficient condition necessary for ensuring equilibrium well-posedness across all information environments.

Using the properties of Pettis integrals, some salient features of an equilibrium can be revealed through the linear best-response formula.
Specifically, by using the Fubini property and LLN for Pettis integrals, we derive two functional equations that are necessarily satisfied by the first and second moments of any joint distribution over equilibrium actions and the state.
These equations extend the moment restrictions found in \cite{bm2013} to the present setting of heterogeneous agents, irrespective of whether signals and the state are Gaussian.
Conversely, by assuming a normal distribution for the state, we show that these moment restrictions, coupled with a statistical restriction inherently fulfilled by second moments, are sufficient for candidate functions to be inducible as the mean and covariance of some equilibrium action-state joint distribution.

These restrictions on feasible equilibrium moments carry significant implications for robust predictions of equilibrium outcomes that hold independently of the specification of agents' private information \citep{bm2013,bm2016}.
To explore these implications, we investigate an information design problem in the symmetric LQG setting under the assumption that the designer's preferences hinge only on the mean and covariance of an equilibrium action-state joint distribution.
%\footnote{This type of LQG information design problem is studied, for example, by \cite{bm2013}, \cite{tamura2018}, \cite{miyashita_ui_lqgd}, and \cite{smolin2023yamashita}.}
Such an objective function may arise, for instance, when the designer aims at maximizing agents' welfare.

While agents have symmetric payoff functions and are treated symmetrically by the designer's objective function, they need not be symmetric in terms of private information.
Specifically, we impose no constraints on feasible signal structures, allowing the designer to transmit personalized signals to a continuum of different agents.
This renders the designer's optimization problem high-dimensional and intractable.
However, by leveraging previous results, we can narrow down the designer's choice set to the functions that satisfy the identified moment restrictions.

Simplifying the problem in this manner, we consider a particularly simple class of signal structures, which we name \emph{targeted disclosure}, highlighted by the following selective nature:
The designer fully discloses the state realization to all agents in a given targeted set, while the remaining agents are informed of nothing.
This type of information disclosure can be easily implemented, irrespective of the state distribution, and witnessed frequently in practice, for instance, in the form of exclusive campaign offers in an e-commerce marketplace.
Targeted disclosure can be parametrized solely by a one-dimensional variable representing the measure of the targeted agents.
Consequently, the optimal targeted disclosure can be characterized as a maximizer of a univariate function.
In Proposition~\ref{prop_tg}, we show that the optimal targeted disclosure is globally optimal among all feasible signal structures, despite numerous alternatives that the designer can employ.

As a notable strength shared by our main economic results, we do not a priori restrict candidate variables in consideration.
Specifically, the uniqueness in Proposition \ref{prop_unique} holds within an unrestricted class of strategy profiles, in contrast to several works in the LQG literature that confine attention to the class of strategies that are linear in signals.\footnote{\cite{morris2002shin} address the uniqueness issue in their beauty contest model on the basis of the contraction mapping theorem, but their arguments are incomplete as a relevant residual term is not yet shown to vanish asymptotically. As noted in Footnote 5 of \cite{ap2007}, such convergence can occur by assuming bounds on the action space, but this assumption is incompatible with the linearity of strategies and the unbounded support of signals. \cite{ap2007} derive a linear equilibrium in a model akin to Example \ref{ex_bm} based on the ``matching coefficient'' method, while the obtained equilibrium is not guaranteed to be unique, including non-linear strategies. \cite{ui2013} establish a stronger uniqueness result by using the theorem of \cite{radner1962}, which is applicable when agents are symmetric. In contrast, our proof technique applies to arbitrary asymmetric payoff and information structures.}
Also, the optimality of targeted disclosure in Proposition~\ref{prop_tg} is supported across all signal structures, including non-Gaussian ones.
This sort of strong theoretical robustness is attained with the aid of mathematical tools pertaining to \emph{integral kernels} and \emph{reproducing kernel Hilbert spaces} in functional analysis.
In Appendix~\ref{app_kernel}, we introduce and develop results relevant to these concepts, which would be of independent technical interest, holding potential usefulness in economic applications beyond the contexts of this paper.

The rest of the paper is organized as follows.
In Section~\ref{sec_pettis}, we offer the definition of Pettis integral and some useful results on it.
In Section~\ref{sec_eq}, we present our game-theoretic model and conduct equilibrium analysis.
Building upon this model, in Section~\ref{sec_info}, we analyze an application to information design problems.
In Section~\ref{sec_conc}, we conclude the paper by leaving some final remarks.
Mathematical details and proofs omitted from the main text can be found in the series of appendixes.

%%%%%%%%%%%%%%%%%%%%%%%%%%%%%%%%%%%%%%%%%%%%%
%%%%%%%%%%%%%%%%%%%%%%%%%%%%%%%%%%%%%%%%%%%%%

\section{Pettis Integral} \label{sec_pettis}

We begin this section by introducing the Pettis integral of general processes, which admit outputs in an abstract Hilbert space.
Subsequently, we discuss the Pettis integral of stochastic processes by specifying the output Hilbert space to comprise square-integrable random variables.

\subsection{Pettis Integral of Hilbert-valued Processes} \label{sec_gen}

An \emph{input space} is given as a finite measure space $(T,\Sigma,\nu)$, which is normalized as $\nu(T) = 1$.
For example, $T$ can be the closed interval $[0,1]$ with the Lebesgue measure.
Let $\Sigma_\nu$ be the collection of all $\nu$-measurable subsets of $T$, which is known to form a $\sigma$-algebra on $T$ such that $\Sigma \subseteq \Sigma_\nu$.\footnote{This is due to Carath\'eodory's extension theorem; see Theorem 10.23 of \cite{ab2006}.}
We say that a function $\phi:T \to \R$ is \emph{$\nu$-measurable} if it is a measurable function from $(T,\Sigma_\nu)$ to $(\R,\calB_\R)$, where $\calB_\R$ denotes the Borel $\sigma$-algebra.
Also, we say that $\psi:T^2 \to \R$ is \emph{jointly measurable} if it is a measurable function from $(T^2, \Sigma_\nu \otimes \Sigma_\nu)$ to $(\R, \calB_\R)$, where $\Sigma_\nu \otimes \Sigma_\nu$ is the product $\sigma$-algebra on $T^2$.

An \emph{output space} is given as a Hilbert space $X$, equipped with the inner product $\langle \cdot, \cdot \rangle_X$ and the induced norm $\|\cdot\|_X$.
The subscript ``$_X$'' is often omitted for simplicity.
For example, $X$ can be a space of random variables with finite second moments, as being specified so in Section~\ref{sec_st}.
That $X$ is Hilbertian is needed to prove Proposition \ref{prop_pettis}, while the notion of Pettis integral can be defined in a more general setting.\footnote{See Chapter 11.10 of \cite{ab2006}.}

A \emph{process} is meant by any function $f:T \to X$ that assigns an output $f(t) \in X$ to each input $t \in T$.
Our primary interest lies in the integration of a process $f$ with respect to $\nu$.
In stating the next definition, and throughout the paper, the integral of any real-valued function shall be understood in the sense of Lebesgue.
The integral range will be suppressed when performed over the entire space $T$.

\begin{definition}
A process $f:T \to X$ is \emph{weakly measurable} if the mapping $t \mapsto \langle x ,f(t) \rangle$ is $\nu$-measurable for every $x \in X$.
Moreover, $f$ is \emph{Pettis integrable} if it is weakly measurable, and if there exists $\bar{f} \in X$ such that
\begin{align} \label{def_pettis}
\langle x, \bar{f} \rangle_X = \int \left\langle y,f(t) \right\rangle_X \d \nu(t), \quad \forall x \in X.
\end{align}
In this case, $\bar{f}$ is called the \emph{Pettis integral} of $f$, which is written as $\text{w-}\int f(t) \d\nu(t)$, or simply, $\int f(t) \d\nu(t)$ when there is no risk of confusion. %\footnote{Pettis integral is also called weak integral.}
\end{definition}

Two standard references on weak measurability and Pettis integral are the textbooks, \cite{diestel1977uhl} and \cite{talagrand1984}.
Here, among several properties of weak measurability and Pettis integral, let us recall the following linearity for later reference:
If $f,g:T \to X$ are weakly measurable, then so is the linear combination $\alpha f + \beta g$ for any $\alpha,\beta \in \R$.
Moreover, if $f$ and $g$ are Pettis integrable, then so is $\alpha f + \beta g$, and the Pettis integral enjoys linearity,
$$\int (\alpha f + \beta g)(t) \d\nu(t) = \alpha \int f(t) \d\nu(t) + \beta \int g(t) \d\nu(t).$$

There are several known sufficient conditions for Pettis integrability.
For instance, \cite{huff} reports that a process $f$ is Pettis integrable if it is weakly measurable, and if the operator $y \mapsto \langle y, f(\cdot) \rangle$ acts weak-to-weak continuously from $X$ to $\calL_1(\nu)$, the Banach space of $\nu$-integrable real-valued functions.
These conditions may not be very tractable from applied standpoints since they involve arbitrary elements of the dual space, most of which are orthogonal to the process $f$ itself.
In light of this, we provide simpler sufficient conditions, which can be stated solely in terms of the moments of the process $f$, that does not invoke the abstract notion of duality.

\begin{proposition} \label{prop_pettis}
Consider the following conditions on a process $f:T \to X$.
\begin{enumerate}[\rm (P1).]
\item \label{cond_ip}
The mapping $t \mapsto \langle f(s),f(t) \rangle$ is $\nu$-measurable for every $s \in T$.
\item \label{cond_n}
The mapping $t \mapsto \|f(t)\|$ is $\nu$-measurable and $\int \|f(t)\| \d \nu (t) < \infty$.
\end{enumerate}
If (and only if) $f$ satisfies {\rm (P1)}, it is weakly measurable.
If, in addition, $f$ satisfies {\rm (P2)}, then it is Pettis integrable.
\end{proposition}

\begin{remark}
By the symmetry of inner products, (P1) is equivalent to saying that the bivariate mapping $(s,t) \mapsto \langle f(s),f(t) \rangle$ is separately measurable in the standard sense; see Definition 4.47 of \cite{ab2006}.
\end{remark}

It is evident that weak measurability implies (P1), as we can substitute $f(s)$ for $x$ in Definition~\ref{def_pettis}.
Therefore, an important part of Proposition~\ref{prop_pettis} is the converse implication.
We establish this by slightly broadening the arguments outlined in \cite{al1995}, which shows the weak measurability of i.i.d.\ random variables by using orthogonal decomposition in a Hilbert space.\footnote{See Footnote 7 and Section A.1 of \cite{al1995}.}
In this regard, it is a needed prerequisite for the proposition that a process admits values in a Hilbert space, although this assumption is rather innocuous as long as we consider, for instance, a space of square-integrable random variables as in Section~\ref{sec_st} and beyond.

A process is said to be strongly measurable when it can be expressed as the norm-limit of a sequence of simple processes, each of which admits at most finitely many values.
For any strongly measurable process, (P2) is known to be necessary and sufficient for the process to be Bochner integrable, an integral notion stronger than Pettis; see Theorem 2 in p.\ 45 of \cite{diestel1977uhl}.
On the other hand, (P2) is sufficient but not necessary for a weakly measurable process to be Pettis integrable, as illustrated by Example~\ref{ex_nonsep} in Appendix \ref{app_pettis}.\footnote{This example is a variant of Birkhoff's example, showing that a continuum of mutually uncorrelated random variables are Pettis-integrated to zero no matter what values are taken for each individual variance; see Example 5 in p.\ 43 of \cite{diestel1977uhl}.}
This example also implies that a Pettis integrable process $f$ can fail to entail jointly measurable inner products $(s,t) \mapsto \langle f(s),f(t) \rangle$ or $\nu$-measurable norms $t \mapsto \|f(t)\|$.
While one can interpret these as indicating the wide-applicability of the integral notion, it would be useful to employ stronger conditions in applications.
Specifically, in our latter game-theoretic analysis, we will require primitives to meet the following strengthened conditions.

\begin{enumerate}[\rm (Q1).]
\item The mapping $(s,t) \mapsto \langle f(s),f(t) \rangle$ is jointly measurable.
\item The mapping $t \mapsto \|f(t)\|$ is $\nu$-measurable and $\int \|f(t)\|^2 \d \nu(t) < \infty$.
\end{enumerate}

Notice that  (Q1) implies (P1) and weak measurability since any jointly measurable function is separately measurable by Theorem 4.48 of \cite{ab2006}.
Additionally, we can show that (Q1) implies that $t \mapsto \|f(t)\|$ is $\nu$-measurable, which is exactly the first part of (Q2).\footnote{To see this, notice that $t \mapsto \|f(t)\|^2$ is written as the composition of two measurable mappings, $t \mapsto (t,t)$ and $(s,t) \mapsto \langle f(s),f(t) \rangle$, acting $T \to T^2$ and $T^2 \to \R$, respectively. Then, the desired conclusion follows from Lemma 4.22 of \cite{ab2006}.}

%%%%%%%%%%%%%%%%%%%%%%%%%%%%%%%%%%%%%%%%%%%%%
%%%%%%%%%%%%%%%%%%%%%%%%%%%%%%%%%%%%%%%%%%%%%

\subsection{Pettis Integral of Stochastic Processes} \label{sec_st}

In the remainder of Section~\ref{sec_pettis}, we specify each output $x \in X$ as a random variable with a finite second moment.
Specifically, let $(\Omega, \Pi, \P)$ be an arbitrary probability space, and let $X$ be a collection of random variables $x: \Omega \to \R$ such that $\E \qty|x|^2 = \int |x(\omega)|^2 \d \P(\omega) < \infty$.
The inner product between random variables $x,y \in X$ is defined as
\begin{align*}
\langle x,y \rangle_X \coloneqq \E \qty[xy] = \int_\Omega x(\omega) y(\omega) \d \P(\omega).
\end{align*}
For any $x,y \in X$, we denote the covariance by $\Cov \qty[x,y] = \E \qty[xy] - \E \qty[x]\E \qty[y]$, the variance by $\Var\qty[x] = \Cov\qty[x,x]$, the standard deviation by $\Sd\qty[x] = \Var\qty[x]^{1/2}$, and the correlation coefficient by $\Corr\qty[x,y] =\Cov\qty[x,y]/\Sd\qty[x]\Sd\qty[y]$.

In this probabilistic context, a \emph{(stochastic) process} is meant by a function $f:T \to X$ that assigns a random variable $f(t) \in X$ to each input $t \in T$.
As an immediate implication of Proposition \ref{prop_pettis}, a stochastic process is Pettis integrable if it has measurable covariance and integrable mean and standard deviation.

\begin{corollary} \label{cor_rv}
Let $f$ be a stochastic process such that $\Cov \qty[f(s),f(\cdot)]$ is $\nu$-measurable for every $s \in T$, and that $\E[f(\cdot)]$ and $\Sd \qty[f(\cdot)]$ are $\nu$-integrable.
Then, $f$ is Pettis integrable.
\end{corollary}

\begin{comment}
\begin{proof}
If $f$ has zero mean, then $\langle f(s),f(t) \rangle = \Cov\qty[f(s),f(t)]$ and $\|f(t)\| = \Sd\qty[f(t)]$, from which (P1) and (P2) are satisfied under the assumptions of the corollary.
Hence, by Proposition~\ref{prop_pettis}, for any general process $f$, its unbiased part $(f-\E\qty[f])$ is Pettis integrable if $\Cov\qty[f(s),f(\cdot)]$ is measurable and $\Sd\qty[f(\cdot)]$ is integrable.
Moreover, since $f = (f-\E f) + \E f$, the process $f$ itself is Pettis integrable if, in addition, $\E \qty[f]$ is integrable.
\end{proof}
\end{comment}

\setcounter{example}{0}
\begin{example}[Continued]
By referring to the result of \cite{ui2013}, Proposition 4 of \cite{bm2013} reports that an equilibrium uniquely exists if $r < 1$.
This equilibrium is symmetric across all agents and characterized as a linear function of signals such that $a_i = \alpha_0^* + \alpha_x^* x_i + \alpha_y^* y$,
where the coefficients $(\alpha_0^*,\alpha_x^*,\alpha_y^*)$ are determined by $(r,s,k,\mu_\theta,\sigma^2_\theta,\sigma^2_x,\sigma^2_y)$.
We can verify that this equilibrium strategy profile is Pettis integrable as follows:
Firstly, the mean is computed as $\E \qty[a_i] = \alpha_0^* + (\alpha_x^* + \alpha_y^*) \mu_\theta$, which is constant across agents, and thus, it is integrable.
Moreover, the covariance is computed as
\begin{align*}
\Cov \qty[a_i, a_j] = \qty(\alpha_x^* + \alpha_y^*)^2 \sigma^2_\theta + \alpha_y^{*2} \sigma_y^2 + \alpha_x^{*2} \sigma_x^2 \cdot \1_{\{i=j\}},
\end{align*}
where $\1_E$ denotes the indicator function of event $E$.
Notice that $\Cov \qty[a_i, a_j]$ is given as the step function that admits constant values on the on-diagonal and off-diagonal subsets of $[0,1]^2$, respectively.
Thus, it is measurable and $\Sd \qty[a_i]$ is integrable, from which Corollary~\ref{cor_rv} confirms that $\{a_i\}_{i \in [0,1]}$ is Pettis integrable.
\end{example}

The definition of Pettis integrals reminds us of a Fubini-like property that justifies interchanging the order of expectation over $\omega \in \Omega$ and integration over $t \in T$.
Indeed, by taking a non-zero constant random variable for $x$ in (\ref{def_pettis}), we confirm the following:
\begin{align} \label{fubini_un}
\E \qty[\int f(t) \d\nu(t)] = \int \E \qty[f(t)] \d\nu(t).
\end{align}

In relation to this formula, we provide two results pertaining to the interchange of probabilistic operations and integration.
The next result justifies the interchange of (co)variance and integration.

\begin{proposition} \label{prop_var}
For any Pettis integrable stochastic process $f$, it holds that
\begin{align} \label{fubini_cov}
\Cov \qty[x, \int f(t) \d\nu(t)] = \int \Cov\qty[x,f(t)] \d \nu(t), \quad \forall x \in X.
\end{align}
In addition, if $s \mapsto \int \Cov \qty[f(s),f(t)] \d \nu(t)$ defines a $\nu$-integrable function, then
\begin{align} \label{fubini_var}
\Var \qty[\int f(t) \d\nu(t)] = \int \int \Cov\qty[f(s),f(t)] \d\nu(s) \d\nu(t).
\end{align}
\end{proposition}

This proposition yields an important implication when $f(t)$ are pairwise uncorrelated and have a common mean across $t$.
In this case, the right-hand side of (\ref{fubini_var}) becomes $0$, suggesting that the Pettis integral coincides with the common mean in mean square; this is exactly the Pettis integral version of LLN, known as Theorem 3 of \cite{uhlig1996}.
However, as cautioned in pp.\ 551--552 of \cite{khan1999sun}, we should be careful when interpreting this result because its probabilistic content is quite different from the classical LLN, stating that a sample average of countably many i.i.d.\ random variables converges to the mean for almost every sample path.
This issue is elaborated in Section~\ref{sec_conc} by referring to their discussion in more detail.

\setcounter{example}{0}
\begin{example}[Continued]
As argued in \cite{ap2007}, two pivotal variables in the welfare analysis of LQG games are \emph{volatility} and \emph{dispersion}, defined respectively as the variance of aggregated action $V \coloneqq \Var\qty[A]$ and that of the idiosyncratic difference $D \coloneqq \Var \qty[A - a_i]$.
\cite{ui2015} show that the expected welfare in any symmetric LQG game can be expressed as a linear combination of $V$ and $D$, thus, calculating these variables holds particular importance for welfare evaluations.
Given any symmetric strategy profile, by relying on the intuition that idiosyncratic parts of individual actions cancel out through aggregation, \cite{bm2013} argue that the volatility and dispersion are computed as follows:
\begin{align*}
V = \Cov \qty[a_i, a_j] \quad \text{and} \quad D = \Var \qty[a_i] - \Cov \qty[a_i, a_j],
\end{align*}
where $i,j \in [0,1]$ are any distinct pair of representative agents.
These formulae are readily derived as formal consequences of Proposition \ref{prop_var}, demonstrating that this kind of LLN can be derived directly from the definition of Pettis integral.
\end{example}

Next, we develop a ``conditional'' version of the Fubini formula \eqref{fubini_un} that arises when we replace the unconditional expectation by the conditional expectation with respect to a sub-$\sigma$-algebra of $\Pi$.

\begin{proposition} \label{prop_cond}
Let $\hat{\Pi} \subseteq \Pi$ be a sub-$\sigma$-algebra, $f$ be a Pettis integrable stochastic process, and $\hat{f} \equiv \E [f \mid \hat{\Pi}]$.
If $\hat{f}$ is Pettis integrable, then we have
\begin{align} \label{fubini_cond}
\P \qty( \E\qty[\int f(t) \d\nu(t) \mid \hat{\Pi}] = \int \E \qty[f(t) \mid \hat{\Pi}] \d\nu(t) ) = 1.
\end{align}
In particular, if $f$ satisfies {\rm (P2)} and $\hat{f}$ satisfies {\rm (P1)}, then $\hat{f}$ is Pettis integrable.
\end{proposition}

According to Proposition~\ref{prop_cond}, the interchange of conditional expectation and integration is justified if the ``conditional'' stochastic process $\hat{f}$ is Pettis integrable.
To assess the Pettis integrability of $\hat{f}$, we remark that $\E [\hat{f}]=\E [f]$ and $\Sd [\hat{f}] \le \Sd [f]$ hold by the standard properties of random variables.
Hence, given that the initial process $f$ meets all assumptions of Corollary~\ref{cor_rv}, our remaining task would be to check if the covariance function of $\hat{f}$ is measurable.
In general, it is not easy to evaluate the covariance of conditional expectations, the task becomes simpler if our focus is restricted to Gaussian processes, as illustrated in the next section.

\begin{comment}
To illustrate, suppose $\hat{\Pi}$ is generated by some random vector $x = (x_1,\ldots,x_n)$ with a non-singular covariance matrix $\Var \qty[x] \in \R^{n \times n}$.
In addition, let $x$ and $f$ together form a \emph{Gaussian process}, namely, meaning any finite selection from the collection $\{x_1,\ldots,x_n\} \cup \{f(t)\}_{t \in T}$ is jointly normally distributed.
Then, according to the conditional Gaussian formula, we have
\begin{align*}
\Cov \qty[\hat{f}(s),\hat{f}(t)] = \Cov\qty[f(s),x]\Var \qty[x]^{-1} \Cov\qty[f(t),x].
\end{align*}
This expression is jointly measurable in $(s,t)$ if a vector-valued function $\Cov\qty[f(t),x]$ is measurable in $t$.
Therefore, the conditional Fubini formula holds if the conditioning random vector $x$ correlates with the process $f$ in a measurable manner.
\end{comment}

%%%%%%%%%%%%%%%%%%%%%%%%%%%%%%%%%%%%%%%%%%%%%
%%%%%%%%%%%%%%%%%%%%%%%%%%%%%%%%%%%%%%%%%%%%%

\subsection{Pettis Integral of Gaussian Processes}
\label{sec_gauss}

In many economic applications, including Example \ref{ex_bm}, we focus on a continuum of normally distributed random variables, which can be formalized by the concept of Gaussian processes.
Formally, a stochastic process $f: T \to X$ is called a \emph{Gaussian process} if any finite selection from the collection of random variables $\{f(t)\}_{t \in T}$ is jointly normally distributed.
By normality, the joint distribution of $f$ can be summarized by the associated mean and covariance functions,
\begin{align*}
\mu(t) \coloneqq \E \qty[f(t)] \quad \text{and} \quad \sigma(s,t) \coloneqq \Cov \qty[f(s),f(t)], \quad \forall s,t \in T.
\end{align*}
By definition, $\sigma$ necessarily satisfies the statistical property of ``positive semidefiniteness,'' which is formally introduced as Definition~\ref{def_psd} in Appendix~\ref{app_kernel}.
Moreover, Theorem 12.1.3 of \cite{dudley} establishes the converse by showing that if $\sigma$ satisfies positive semidefiniteness, then there exists a Gaussian process that has $\sigma$ exactly as its covariance function.
No restriction is needed for $\mu$.

Proposition~\ref{prop_pettis} or Corollary~\ref{cor_rv} can readily provide sufficient conditions on $\mu$ and $\sigma$ for a given Gaussian process to be Pettis integrable.
Moreover, Proposition~\ref{prop_cond} can be slightly enhanced when we focus on Gaussian processes.

\begin{corollary} \label{cor_fubini}
Let $f:T \to X$ be a Gaussian process that satisfies {\rm (Q1)} and {\rm (Q2)}.
Then, for any $n \in \N$ and $t_1,\ldots,t_n \in T$, the process $\hat{f} \equiv \E \qty[f \mid f(t_1),\ldots,f(t_n)]$ is Pettis integrable, and it holds that
\begin{align} \label{cond_gauss}
\P\qty(\E \qty[\int f(s) \d \nu(s) \mid f(t_1),\ldots,f(t_n)] = \int \E \qty[f(s) \mid f(t_1),\ldots,f(t_n)] \d \nu(s)) = 1.
\end{align}
\end{corollary}

According to this corollary, a Gaussian process enjoys the conditional Fubini formula~\eqref{cond_gauss}, given that the process satisfies our regularity conditions, (Q1) and (Q2).
To verify this result, we can use the conditional Gaussian formula to observe that
\begin{align*}
\Cov \qty[\hat{f}(s), \hat{f}(t)] = \mqty[\sigma(s,t_1) \\ \vdots \\ \sigma(s,t_n)] ^\top \mqty[\sigma(t_1,t_1) & \cdots & \sigma(t_1,t_n) \\ \vdots & \ddots & \vdots \\ \sigma(t_n,t_1) & \cdots & \sigma(t_n,t_n)]^{-1} \mqty[\sigma(t,t_1) \\ \vdots \\ \sigma(t,t_n)]
\end{align*}
is a jointly measurable function in terms of $(s,t)$, provided that so is $\sigma$.
Therefore, by our proceeding discussion after Proposition~\ref{prop_cond}, we can conclude that the conditional Fubini formula holds true.

\setcounter{example}{0}
\begin{example}[Continued]
A common way to get a closed form of an equilibrium is based on matching coefficients.
Namely, we guess a (symmetric) equilibrium strategy $a_i = \alpha_0 + \alpha_x x_i + \alpha_y y$, which is linear in signals, and substitute it into the best-response formula \eqref{br_bm} to get the expressions of $(\alpha_0,\alpha_x,\alpha_y)$ in terms of parameters.
In doing so, we need to calculate each agent's conditional expectation of the aggregated action $A$, which is facilitated by postulating the exchangeability of conditional expectation and integration,
\begin{align*}
\E \qty[A \mid x_i,\, y] = \int_0^1 \E \qty[a_j \mid x_i,\, y] \d j.
\end{align*}
By means of Proposition~\ref{prop_cond}, the interchange is justified if the conditional process $j \mapsto \E \qty[a_j \mid x_i,\, y]$ is Pettis integrable.
Indeed, by the linearity of $a_j$ in $(x_j,y)$ and the normality of information, one can see that $\E \qty[a_j \mid x_i,\, y]$ is expressed as a linear combination of $x_i$ and $y$, which has identical coefficients across all agents.
Consequently, the conditional process maintains the needed conditions for Pettis integrability.
\hfill $\triangle$
\end{example}

%%%%%%%%%%%%%%%%%%%%%%%%%%%%%%%%%%%%%%%%%%%%%
%%%%%%%%%%%%%%%%%%%%%%%%%%%%%%%%%%%%%%%%%%%%%

\section{Equilibrium Analysis}
\label{sec_eq}

Now, we apply our mathematical study in the last section to the formal equilibrium analysis of large population games with incomplete information.
To this end, we consider generalizing the model of Example~\ref{ex_bm} in regard with two important aspects.
First, we no longer assume that agents have an identical payoff function by allowing for parameters to depend on their indexes.
Second, by dropping the symmetry and normality, we generalize an information environment.

\begin{example}
Let us generalize Example~\ref{ex_bm} to accommodate asymmetry across agents in terms of both payoffs and information.
Specifically, we postulate that agent $i$'s best-response strategy is now given as
\begin{align} \label{br2}
a_i = \E \qty[\int_0^1 r_{ij} a_j \d j \mid x_{i1},\ldots,x_{i,n}] + s_i \E \qty[\theta \mid x_{i,1},\ldots,x_{i,n}] + k_i,
\end{align}
where $r_{ij}, s_i, k_i \in \R$ are the parameters that can vary across agents and $x_{i,1},\ldots,x_{i,n}$ are $n$ private signals observed by agent $i$.

To interpret, $r_{ij}$ indicates how agent $i$ would adjust her strategy in response to the marginal change of agent $j$'s strategy.
For example, in network games, these can be specified as $r_{ij} = \bar{r} \cdot \1_{\{(i,j) \in G\}}$, where $\bar{r} > 0$ represents the degree of peer effects, and $G \subseteq T^2$ is an arbitrary undirected graph that represents the existence of links between agents.
The dependency of $s_i$ and $k_i$ on indexes allows different agents to exhibit different base action levels and responses to the state.
Moreover, each agent's signal can be heterogeneous in terms of the correlations with the state or other agent's signal, whereas the model does not preclude public signals.
For instance, \cite{miyashita_id} introduces the following \emph{Gaussian information structure}: Each agent's signal is given as a random vector $x_i:\Omega \to \R^n$ such that the collection of random variables $\{\theta\} \cup \{x_{i,1},\ldots,x_{i,n}\}_{t \in T}$ constitutes a Gaussian process.
Note that the signal-state joint distribution of Example~\ref{ex_bm} is a special type of Gaussian information environment, where $x_1(t)$ are conditionally i.i.d.\ and $x_2(t)$ are common across agents.
\hfill $\triangle$
\end{example}

\subsection{Setup}
\label{sec_eq_model}

We now interpret the input space $T$ as the set of agents.
For example, it can be the unit interval $ [0,1]$ with the Lebesgue measure or any finite set $\{1,\ldots,n\}$ with the uniform mass function.
The output space $X$ continues to be the space of square-integrable random variables, modeled on some probability space $(\Omega,\Pi,\P)$, on which the incomplete information game takes place.

There is a payoff-relevant state for each agent $t$, given as an exogenous random variable $\theta(t) \in X$.
While we allow for $\theta(t)$ to be idiosyncratic across agents, the case of the common state can be captured by taking all $\theta(t)$ as an identical random variable.
Prior to choosing actions, each agent $t$ receives a private signal $x(t)$, given as a random variable that admits values in some measurable space.
Denote by $\E_t \qty[\cdot] = \E \qty[\cdot \mid x(t)]$ the agent $t$'s conditional expectation operator.
We refer to the collection of random variables $\{\theta(t),x(t)\}_{t \in T}$ as an \emph{information environment}.

An action is taken from $\R$ as a function of one's signal realization.
Specifically, agent $t$'s \emph{strategy} is given as an $x(t)$-measurable random variable $f(t):\Omega \to \R$ such that $\E\qty|f(t)|^2 < \infty$.
We regard two strategies $f(t)$ and $f'(t)$ as the same if $\E \qty|f(t)-f'(t)|^2 = 0$.
Assume that the agent's payoff depends on her own action choice $f(t)$ and state $\theta(t)$, as well as the aggregated action $F(t)$ of others.
Unlike Example~\ref{ex_bm}, we no longer assume that the way of aggregation is common across agents.
Instead, the agent $t$'s \emph{local aggregate} is given as a weighted integral,
\begin{align*}
F(t) = \int R(t,t') f(t') \d \nu(t'),
\end{align*}
where $R:T^2 \to \R$ is a bivariate real-valued function.
Here, $F(t)$ is interpreted as the Pettis integral of random variables $\{f(t')\}_{t' \in T}$.
For the well-definedness of it, we will impose some assumptions on $R$ and restrictions on $f$.

By assuming that agent $t$'s payoff is quadratically dependent on $f(t)$, $F(t)$, and $\theta(t)$, we posit that she best responds to others by adopting a linear strategy in terms of her best estimates regarding the local aggregate and the state.\footnote{The linear best-response strategy is seen as a reduced-form assumption that agents are maximizing quadratic payoffs,
\begin{align*}
-\frac{1}{2} f(t)^2 + f(t) \cdot (F(t)+\theta(t)) + U,
\end{align*}
conditional on their private information. Here, $U$ stands for any component of utility that is determined independently of $f(t)$.}
Specifically, we postulate that the best-response formula takes the form,
\begin{align} \label{br}
f(t) = \E_t \qty[F(t)] + \E_t \qty[\theta(t)].
\end{align}
Notice that the flexibility of $R(s,t)$ and $\theta(t)$ allows us to formulate the best response formula as the sum of the best estimates rather than as linear combinations.
We refer to the function $R$ as a \emph{payoff structure}.

Our incomplete information game can be summarized as a profile $\left\langle \{\theta(t), x(t)\}_{t \in T}, R \right\rangle$ of an information environment and a payoff structure.
In later analysis, we treat signals as endogenously chosen variables while the state distribution and payoff structure are exogenously given.
Motivated by this, we refer to $\langle \{\theta(t)\}_{t \in T}, R \rangle$ as a \emph{basic game} and $\{x(t)\}_{t \in T}$ as a \emph{signal structure}.
Throughout, we impose the following mild regularity conditions on basic games:
First, when viewed as a process, we assume that the state $\theta:T \to X$ satisfies (Q1) and (Q2).
Second, a payoff structure $R:T^2 \to \R$ is given as a jointly measurable function such that
\begin{align*}
\int |R(s,t)|^2 \d \nu(t) < \infty,\, \forall s \in T \quad \text{and} \quad \int \int |R(s,t)|^2 \d \nu(s) \d \nu(t) < \infty.
\end{align*}
Notice that both conditions are trivially satisfied, for example, if the set of agents $T$ is at most countable or if all agents are homogeneous so that $\theta(t)$ and $R(s,t)$ do not depend on their indexes.

In principle, a \emph{strategy profile} can be given as any process $f: T \to X$ such that $f(t)$ is $x(t)$-measurable for every $t \in T$.
In order to assure that the local aggregates are well-defined, however, we employ (Q1) and (Q2) to stipulate the class of regular strategy profiles.
By accommodating these regularity requirements as a part of equilibrium condition, we say that a regular strategy profile constitutes an equilibrium when every agent takes her best-response strategy.

\begin{definition} \label{def_st}
A strategy profile $f$ is \emph{regular} if it satisfies (Q1) and (Q2).
Moreover, a regular strategy profile $f$ constitutes an \emph{equilibrium} if \eqref{br} holds for every $t \in T$.
\end{definition}

Let us introduce a few properties of payoff structures, which hold pivotal roles in characterizing the uniqueness property of the linear best response model.
While some mathematical notions are needed to formally state these properties, all of which can be found in Appendix~\ref{app_kernel}.
In this appendix, we also present some known facts and new lemmas that are of relevance to our economic analysis.

We shall identify a payoff structure $R:T^2 \to \R$ with an integral operator $\mathbf{R}$, acting for real-valued functions $\phi:T \to \R$, as follows:
\begin{align} \label{def_R}
\mathbf{R}\phi \coloneqq \int R(\cdot, t) \phi(t) \d \nu(t).
\end{align}
Mathematically, a bivariate function $R:T^2 \to \R$ is known with the name of \emph{(integral) kernels}, which can be seen as natural counterparts of matrixes in infinite dimensional analysis.
Several common terminologies for matrixes can be extended to kernels, allowing us to introduce the following two conditions on payoff structures.

\begin{enumerate}[\rm (R1).]
\item The numerical range of $\mathbf{R}$ is contained in $(-\infty,1)$.
\item The set of (real) eigenvalues $\mathbf{R}$ is contained in $(-\infty,1)$.
\end{enumerate}

At the moment, we remark that these conditions are equivalent to each other for any \emph{undirected} payoff structure such that $R(s,t) = R(t,s)$ holds for all $s,t \in T$.
This type of payoff structure is often considered in the context of network games under the premise that strategic interactions between agents are determined solely by the existence of undirected peer links connecting them.
In general, taking directed payoff structures into consideration, (R1) is strictly stronger than (R2), as Example~\ref{ex_uni} in Appendix~\ref{app_eq} illustrates the gap between these conditions.

Roughly speaking, (R1) and (R2) regulate the magnitude and direction of strategic interactions in different yet similar ways.\footnote{A relevant condition is that the spectral radius of $\mathbf{R}$ is less than 1, which is often employed in the network game literature, e.g., Proposition 3.5 of \cite{jackson2015zenou}. This condition is stronger than (R1) and (R2), as it regulates the absolute values of eigenvalues. In contrast, (R1) and (R2) impose upper bounds but do not entail lower bounds.}
These conditions become easy to interpret when $R(s,t) = r$ is constant across all agents.
In this case, (R1) is equivalent to (R2), and both are reduced to the one-dimensional parametric condition that $r < 1$.\footnote{This is a commonly adopted condition in symmetric LQG games, e.g., \cite{ap2007}, \cite{bm2013}, and \cite{ui2015}. Under this condition, \cite{ui2013} establishes the unique existence of an equilibrium in a symmetric LQG game.}
Economically, this can be interpreted as precluding the case of large strategic complementarities (i.e., $r \ge 1$), whereas small strategic interactions (i.e., $|r| < 1$) or strategic substitutes (i.e., $r < 0$) are allowed to exist.
We can extend these interpretations to the present asymmetric setting by clarifying relevant sufficient conditions for (R1).

\begin{lemma} \label{lem_r}
The following hold true:
\begin{enumerate}[\rm i).]
\item {\rm (R1)} implies {\rm (R2)}, while the converse is true if $R$ is undirected.
\item {\rm (R1)} is satisfied if $\sup_{s,t \in T} |R(s,t)| < 1$.
Alternatively, {\rm (R1)} is satisfied if $R$ takes a multiplicatively separable form, $R(s,t) = rq(s)q(t)$, where $r < 0$ is any negative number and $q:T \to \R$ is any square-integrable function.
\end{enumerate}
\end{lemma}

Together with the mentioned relationship between (R1) and (R2), the lemma provides two sufficient conditions for (R1).
The first condition says that the game entails strategic interactions of bounded magnitude.
The second condition posits the negativity of strategic interactions by requiring that $R(s,t)$ is multiplicatively separable into the terms specific to each agent, $q(s)$ and $q(t)$, and a fixed negative coefficient $r$.

%%%%%%%%%%%%%%%%%%%%%%%%%%%%%%%%%%%%%%%%%%%%%
%%%%%%%%%%%%%%%%%%%%%%%%%%%%%%%%%%%%%%%%%%%%%

\subsection{Moment Restrictions}
\label{sec_moment}

We begin the analysis by identifying constraints imposed on any equilibrium of the linear best response model.
The next proposition offers two functional equations, which are {\it necessarily} satisfied by the first and second moments of any equilibrium action-state joint distribution.
These equations generalize the obedience conditions of \cite{bm2013} to the present asymmetric setting, irrespective of whether an information environment is Gaussian.

\begin{proposition} \label{prop_moment}
If $f$ is an equilibrium in $\langle \{\theta(t),x(t)\}_{t \in T}, R \rangle$, then it satisfies the following moment restrictions for all $t \in T$:
\begin{gather}
\E \qty[f(t)] = \int R(t,t') \E \qty[f(t')] \d \nu(t') + \E \qty[\theta(t)], \label{moment1} \\
\Var \qty[f(t)] = \int R(t,t') \Cov \qty[f(t),f(t')] \d \nu(t') + \Cov \qty[f(t), \theta(t)]. \label{moment2}
\end{gather}
In particular, under {\rm (R2)}, the first equation \eqref{moment1} uniquely pins down $\E\qty[f(t)]$ for all $t \in T$ through $R(\cdot,\cdot)$ and $\E \qty[\theta(\cdot)]$.
\end{proposition}

The first equation \eqref{moment1} demands that the mean of each agent's action is expressed as the weighted aggregate of other agents' expected actions, plus the expected state.
Consequently, given any payoff structure that satisfies (R2), the mean action is uniquely determined for every agents across all equilibria that would arise in different information environments, as long as the marginal over the state is fixed.
In fact, to attain this uniqueness property, we can replace (R2) by a much weaker condition, requiring that the integral operator $\mathbf{R}$ does not admit $1$ as its eigenvalue, which can be satisfied by a generic payoff structure.

The second equation \eqref{moment2} restricts the (co)variance of an equilibrium action-state distribution by forcing each agent's action variance to be expressed as the aggregation of the correlations with the others agents' actions and the state.
When the game is symmetric, by postulating that the equilibrium strategies are symmetric as well, the equation can be simplified as
\begin{align} \label{moment_bm}
\Sd \qty[f(t)] = \frac{\Corr \qty[f(t),\theta]}{1-r \Corr \qty[f(t),f(t')]} \cdot \Sd \qty[\theta],
\end{align}
where $t,t' \in T$ are any distinct representative agents.
This is exactly the second moment restriction in Proposition 1 of \cite{bm2013}.
As pointed out by them, the second moment is restricted less stringently than the first moment, as there are three free variables---$\Sd \qty[f(t)]$, $\Corr \qty[f(t),f(t')]$, and $\Corr \qty[f(t), \theta]$---that are abide by a one-dimensional equation \eqref{moment_bm} and another one-dimensional inequality, stemming from a statistical requirement.
In the asymmetric case, the latter statistical requirement can be extended to a multi-dimensional positive semidefiniteness constraint, which is shortly formalized as in Definition~\ref{def_moment}.

In the remainder of Section~\ref{sec_moment}, we assume that all agents are concerned about the common state $\theta \equiv \theta(t)$, which follows some normal distribution.
Under this additional assumption, we can show that the moment restrictions, together with the positive semidefiniteness constraint, are \emph{sufficient} for any candidate moment functions to be induced in an equilibrium of some information environment.
To this end, since the equilibrium mean is uniquely pinned down by the first moment restriction, we just focus on the inducibility of the second moment.

\begin{definition} \label{def_moment}
Let $\xi:T^2 \to \R$ and $\zeta:T \to \R$ be square-integrable functions.
Given any payoff structure $R$, a pair of functions $(\xi,\zeta)$ is called an \emph{equilibrium moment (under $R$)} if the following conditions are satisfied:
\begin{itemize}
\item \emph{Obedience}: For any $t \in T$,
\begin{align*}
\xi(t,t) = \int R(t,t')\xi(t,t') \d\nu(t') + \zeta(t).
\end{align*}
\item \emph{Positivity}: For any $n \in \N$ and $t_1,\ldots,t_n$,
\begin{align*}
\mqty[\xi(t_1,t_1) & \cdots & \xi(t_1,t_n) & \zeta(t_1) \\ \vdots & \ddots & \vdots & \vdots \\ \xi(t_n,t_1) & \cdots & \xi(t_n,t_n) & \zeta(t_n) \\ \zeta(t_1) & \cdots & \zeta(t_n) & \Var\qty[\theta]] \succeq O.
\end{align*}
\end{itemize}
\end{definition}

\begin{remark}
For a square matrix $M$, we write as $M \succeq O$ when $M$ is symmetric and positive semidefinite.
Though Definition~\ref{def_moment} itself can be extended to the case of idiosyncratic states by considering a bivariate function $\zeta: T^2 \to \R$, the current proof of Proposition~\ref{prop_moment2} assumes the common state.
\end{remark}

As will be clear soon, $\xi(t,t')$ can be interpreted as representing the covariance between equilibrium actions of two agents $t$ and $t'$.
Similarly, we can interpret $\zeta(t)$ as the covariance between $t$'s action and the common state.
By Proposition~\ref{prop_moment}, for a pair of functions $(\xi,\zeta)$ to be valid as the equilibrium action-state covariance, they necessarily satisfies the second moment restriction, stated as the obedience condition in Definition~\ref{def_moment}.
In addition, they need to maintain the positivity condition to serve as the representation of the covariance of a collection of random variables $\{\theta\} \cup \{f(t)\}_{t \in T}$.

The next proposition shows that for any pair of functions that jointly satisfy obedience and positivity, there exists an information environment that admits an equilibrium, wherein those functions are induced exactly as the covariance of the equilibrium action-state joint distribution.

\begin{proposition} \label{prop_moment2}
Suppose that the common state $\theta$ is normally distributed and $R$ satisfies {\rm (R2)}.
Let $(\xi,\zeta)$ be any equilibrium moment.
Then, there exists a signal structure $\{x(t)\}_{t \in T}$ for which the game $\langle \{\theta(t),x(t)\}_{t \in T},r \rangle$ has an equilibrium $f$ such that
\begin{align*}
\quad \Cov \qty[f(t),f(t')] = \xi(t,t') \quad \text{and} \quad \Cov \qty[f(t), \theta] = \zeta(t), \quad \forall t,t' \in T.
\end{align*}
In particular, we can let $x(t) = f(t)$ for all $t \in T$, each of which is a one-dimensional normal random variable.
\end{proposition}

The proof of this proposition is constructive:
Given an equilibrium moment $(\xi,\zeta)$, we let each agent's signal $x(t)$ be given as a linear combination of the state $\theta$ and an idiosyncratic noise term $\epsilon(t)$, where $\epsilon:T \to X$ can be given as a Gaussian process, which is independently distributed from $\theta$, while $\epsilon(t)$ and $\epsilon (t')$ are allowed to entail non-trivial correlations.
In the linear expression of $x(t)$, the coefficient of $\theta$ is set proportional to $\zeta(t)$ so that $\zeta(t)$ arises as the covariance between $x(t)$ and $\theta$.
Additionally, the covariance function of $\epsilon$ is calibrated in such a way that $\xi(t,t')$ obtains as the covariance between $x(t)$ and $x(t')$.\footnote{The construction bears similarity to canonical information structures in \cite{miyashita_id}.}
The obedience condition ensures that $\{x(t)\}_{t \in T}$ constitutes an equilibrium, while the positivity ensures the existence of the Gaussian process $\epsilon$.
In addition, as each $x(t)$ obtains as a linear combination of normally distributed random variables, the resulting information environment is Gaussian.

\subsection{Uniqueness of Equilibrium}
\label{sec_eq_unique}

Now, we present our main result in this section, which shows that (R1) is sufficient, while (R2) is necessary, for the equilibrium uniqueness to hold in {\it all} information environments.

\begin{proposition} \label{prop_unique}
Suppose that {\rm (R1)} is satisfied.
In any information environment, if $f$ and $g$ are two equilibria such that $(s,t) \mapsto \E \qty[f(s) g(t)]$ is jointly measurable, then we have $\E \qty|f(t)-g(t)|^2 = 0$ for all $t \in T$.
On the other hand, for any basic game $\langle \{\theta(t)\}_{t \in T}, R \rangle $ wherein {\rm (R2)} is violated, there exists a signal structure $\{x(t)\}_{t \in T}$ such that the game $\langle \{\theta(t),x(t)\}_{t \in T},R \rangle$ admits either no equilibrium or a continuum of equilibria.
\end{proposition}

To prove the sufficiency part, we scrutinize the first and second moments of the process $h=f-g$, obtained as the difference between two equilibria.
From \eqref{moment1}, notice that the mean of $h$ must be zero.
Having confirmed this, we then show that $\Sd \qty[h]$ is zero to conclude that $f$ coincides with $g$ in mean square.
To this end, we use \eqref{moment2} to deduce that that the second moment of $h$ satisfies
\begin{align*}
\Sd \qty[h(t)] = \int \underbrace{\qty(R(t,t')\Corr\qty[h(t),h(t')])}_{\eqqcolon \ Q(t,t')} \cdot \Sd \qty[h(t')] \d \nu(t'), \quad \forall t \in T,
\end{align*}
where the joint measurability of $(t,t) \mapsto \E \qty[f(t) g(t')]$, as assumed in Proposition~\ref{prop_unique}, has a technical role in ensuring the standard deviations and correlation coefficients of $h$ to be measurable in agents' indexes.

The above equation suggests that the mapping $t \mapsto \Sd \qty[h(t)]$ must be a solution to the functional equation, expressed as $\phi = \mathbf{Q} \phi$, where $\mathbf{Q}$ is the integral operator whose integral kernel $Q:T^2 \to \R$ is given as the pointwise (or Hadamard) product of two bivariate functions.
Then, our remaining task will be to demonstrate that this operator does not have $1$ as its eigenvalue, thereby concluding that the equation admits only a trivial solution.
To achieve this step, we establish an infinite-dimensional analog of Schur's product theorem (Lemma~\ref{lem_eigen}), which derives an upper bound on the eigenvalues of $\mathbf{Q}$, by utilizing a recent extension of the classical Mercer representation theorem in functional analysis due to \cite{steinwart2012scovel}.
%This may have some intuitive appeal: $\hat{\mathbf{R}}$ is defined through the product of $R(s,t)$, whose eigenvalues are all less than 1 under (R1), and the correlation coefficients of random variables, which are all between $0$ and $1$.

%The necessity part is proved in a constructive manner: We specify each agent's signal as a one-dimensional Gaussian signal $x(t)$ such that \begin{align} \label{gauss_sig} \E \qty[x(t)] = 0 \quad \text{and} \quad \Cov \qty[x(s),x(t)] = \begin{cases} 1 &\text{if } s=t \\ 1/\lambda &\text{if } s\neq t \end{cases}, \end{align} where $\lambda$ is any eigenvalue of $\mathbf{R}$ no less than $1$. Indeed, the existence of such a process $x$ is ensured because $1/\lambda \in (0,1]$ makes the above-defined covariance function positive semidefinite. We then show that whenever there exists at least one equilibrium $f$, an alternative strategy profile $g(t) = f(t) + \phi(t) x(t)$ givens rise to another equilibrium by choosing $\phi$ as any eigenvector associated with $\lambda$. In particular, by the homogeneity of eigenvectors, our construction implies that the cardinality of equilibria is necessarily uncountable whenever multiplicity occurs.

While there is no existence result that holds under as general presumptions as those of Proposition~\ref{prop_unique}, the existence of an equilibrium can be ensured with some additional assumptions.
For instance, \cite{ui2016} and \cite{ozdaglar2019} establish the unique existence of an equilibrium in finite network games by characterizing the equilibrium as a solution to a variational inequality, defined on the Hilbert space of strategy profiles.
However, this proof strategy cannot be directly applied when the set of agents is uncountable, as spaces of Pettis integrable processes do not possess completeness.\footnote{See the survey by \cite{musial2002}.}

Alternatively, in a continuum population LQG model, Theorem 1 of \cite{miyashita_id} establishes that for ``almost every'' combination of a payoff structure and a distribution of Gaussian signals, an equilibrium exists and is unique, while the uniqueness is restricted within the class of linear strategy profiles.\footnote{\cite{lmo2018} obtain similar generic results in linear-quadratic games with finite players.}
It should be remarked that our necessity assertion does not contradict this theorem: The signal structure constructed in our proof simply corresponds to a knife-edge case in which the generic well-posedness result does not apply.
On the other hand, our sufficiency assertion provides a complementary implication to the theorem by establishing a stronger uniqueness property.
Specifically, by combining the two findings, we can derive the following conclusion: In every Gaussian information environment, a linear equilibrium exists for almost every payoff structure, and this equilibrium is unique, including non-linear strategy profiles, provided that the payoff structure satisfies (R1).
%While we predict that the failure of uniqueness is not robust against small perturbations of information structures, given the generic well-posedness, Proposition \ref{prop_unique} indicates that our condition is essentially necessary and sufficient to guarantee the uniqueness in all information environments. In light of these observations, our conditions can be regarded as parsimonious to guarantee the well-posedness of equilibrium.

%%%%%%%%%%%%%%%%%%%%%%%%%%%%%%%%%%%%%%%%%%%%%
%%%%%%%%%%%%%%%%%%%%%%%%%%%%%%%%%%%%%%%%%%%%%

\section{Application to Information Design}
\label{sec_info}

Building upon the model of Section \ref{sec_eq}, we address information design problems \citep{kg2011}, where the information environment is endogenously selected by an extra player of the game, named an \emph{(information) designer}.
The designer has her own objective function, which is quadratically dependent on the equilibrium action-state joint distribution.
Her task is to select an optimal distribution of agents' signals, aiming to maximize the expected objective by anticipating that the agents will play an equilibrium under the chosen signal structure.

To gain tractability, we revert to a symmetric setting like Example~\ref{ex_bm}, where agents are uniformly distributed on the closed interval $T = [0,1]$.
We assume a symmetric payoff structure $R(s,t) \equiv r$ and a common payoff state $\theta(t) \equiv \theta$ for all $s,t \in [0,1]$.
Building upon the previous analysis, we assume $r < 1$ so that the uniqueness of an equilibrium is guaranteed, irrespective of the designer's choice of signal structures.
We normalize the common state $\theta$ to have $\E \qty[\theta] = 0$ and $\Var \qty[\theta] = 1$, as this adjustment does not affect subsequent analyses, while the distribution of $\theta$ need not be normal.
In what follows, we treat a basic game, $r < 1$ and the distribution of $\theta$, as exogenously given.
In contrast, a signal structure, $\{x(t)\}_{t \in T}$, is endogenously chosen by the designer.
%In summary, we make the following assumption throughout this section. \begin{assumption*} Let $r < 1$ and $\theta \sim \calN \qty(0,1)$. \end{assumption*}

\subsection{Designer's Problem}

We postulate that the designer's objective depends on up to the second moment of an equilibrium acton-state joint distribution.
As such, since the mean action levels are invariant across all signal structures, it is without loss of generality for the designer's problem to assume that her objective is determined by the second moment alone.
Then, given any joint distribution of equilibrium actions $\{f(t)\}_{t \in T}$ and the state $\theta$, let us specify the designer's expected objective in the following form:
\begin{align} \label{obj}
u \int_0^1 \int_0^1 \Cov \qty[f(s),f(t)] \d s \d t + v \int_0^1 \Var \qty[f(t)] \d t + w \int_0^1 \Cov \qty[f(t), \theta] \d t.
\end{align}
Here, $u,v,w \in \R$ are given parameters that, respectively, measure the marginal contributions of the correlation between different agents' actions, the variance of an individual action, and the correlation between an individual action and the state.
Given the payoff symmetry on the agents' side, we assume that these parameters do not depend on their indexes.
Hence, all agents are treated equally by the designer.

The designer's problem is to choose a signal structure $\{x(t)\}_{t \in T}$ in such a way of maximizing the expected objective, given as \eqref{obj}, subject to the constraint that $f$ forms an equilibrium in the induced game $\langle \theta, \{x(t)\}_{t \in T}, r \rangle$.\footnote{More precisely, the designer's task also includes constructing the underlying probability space on which $\theta$ and $\{x(t)\}_{t \in T}$ are modeled, while fixing the marginal over $\theta$. Specifically, the designer's choice procedure can be described as follows: Let $(\Omega_\theta,\Pi_\theta,\P_\theta)$ be a given probability space, on which the state $\theta$ is modeled. The designer prepares a probability space $(\Omega_\epsilon,\Pi_\epsilon,\P_\epsilon)$, on which any noise terms are modeled, as well as measure spaces $\calX(t)$, each of which serves as an output space of agent $t$'s signal. The grand probability space is constructed as $(\Omega,\Pi,\P) = (\Omega_\theta \times \Omega_\epsilon, \Pi_\theta \otimes \Pi_\epsilon, \P_\theta \otimes \P_\epsilon)$. Then, a signal structure is given as a collection $\{x(t)\}_{t \in T}$, each of which takes the form $x(t): \Omega \to \calX(t)$.}
Instead of modeling the choice of a continuum of random variables, however, our preceding results simply the problem by deducing an upper bound on the maximum value that the designer can achieve by employing an arbitrary signal structure.
Specifically, Proposition~\ref{prop_moment} implies that any equilibrium necessarily satisfies the moment restriction \eqref{moment2}, as well as the positive semidefiniteness of covariance.
Therefore, we can conclude that the designer's value is \emph{at most} the value of the following optimization problem:
\begin{align} \label{problem}
V^* \coloneqq \quad \sup_{(\xi,\zeta)} \quad \underbrace{u \int_0^1 \int_0^1 \xi(s,t) \d s \d t + v \int_0^1 \xi(t,t) \d t + w \int_0^1 \zeta(t) \d t}_{\eqqcolon \ V(\xi,\zeta)},
\end{align}
where the choice of $(\xi,\zeta)$ is abide by the obedience and positivity constraints in Definition~\ref{def_moment}.

Conversely, if $\theta$ is normally distributed, then Proposition~\ref{prop_moment2} implies that this upper bound on the designer's expected objective is tight and can be achieved by using some Gaussian signal structure.
In fact, even without the normality of $\theta$, the main result in this section reveals that $V^*$ is achievable by a particularly simple information information structure, which selectively discloses the entire state realization only to the targeted set of agents, and nothing else to others.
In what follows, we simply refer to $V^*$ as the designer's value.
A signal structure is said to be \emph{globally optimal} if it induces an equilibrium that delivers $V^*$ to the designer.

%Furthermore, since Proposition~\ref{prop_moment2} assures that any choice of such $(\xi,\zeta)$ is inducible under some signal structure, we can conclude that $V^*$ coincides with the designer's value.\footnote{In particular, by the normality of $\theta$, the proposition implies that $V^*$ is achieved by using a signal structure such that each agent's signal is given as a one-dimensional normal random variable. Our targeted disclosure, indeed, falls in this class of Gaussian signal structures.}

\subsection{Targeted Disclosure}
\label{sec_tg}

Now, we focus on a special kind of signal structures and argue that the optimum within this subclass is globally optimal among all possible signal structures.

\begin{definition}
Given any Lebesgue measurable set of agents $M \subseteq [0,1]$, a signal structure $\{x(t)\}_{t \in T}$ is called the \emph{targeted disclosure to $M$} if
\begin{align*}
x(t) = \begin{cases}
\theta & \text{if} \quad t \in M, \\
\n &\text{if} \quad t \in [0,1] \setminus M,
\end{cases}
\end{align*}
where $\n$ is any deterministic random variable, indicating ``no information.''
\end{definition}

In words, targeted disclosure fully disclose the state realization (and nothing else) to all agents in the ``targeted set,'' while all others are informed of nothing.
These signal structures are simple enough to allow us to establish the existence of an equilibrium from scratch and derives its explicit form.
The next lemma further calculates the equilibrium moment that arises  under such a signal structure and indicates that it depends on the signal structure only through the Lebesgue measure of the targeted set.

\begin{lemma} \label{lem_tg_eq}
Let $M \subseteq [0,1]$ be any set of Lebesgue measure $m$.
There exists a unique equilibrium under the targeted disclosure to $M$, which is given as $f(t)= 0$ for $t \in [0,1] \setminus M$ and  $f(t) = \frac{\theta}{1-rm}$ for $t \in M$.
Consequently, we have
\begin{align*}
\Cov \qty[f(s),f(t)] = \frac{\1_{\qty{s,t \in M}}}{\qty(1-rm)^2}
\quad \text{and} \quad
\Cov \qty[f(t),\theta] = \frac{\1_{\qty{t \in M}}}{1-rm}, \quad \forall s,t \in [0,1].
\end{align*}
\end{lemma}

By this lemma, given any targeted set $M$ of measure $m$, we can calculate the designer's objective as follows:
\begin{align} \label{prob_tg}
V_{\rm tg}(m) \coloneqq
\frac{\alpha m - \beta m^2}{(1-rm)^2} \quad \text{where} \quad \alpha \coloneqq v+w \quad \text{and} \quad \beta \coloneqq rw-u.
\end{align}
Within the subclass of targeted disclosure policies, therefore, a signal structure is optimal if it targets exactly the measure $m^*$ of agents, where $m^*$ is any value that maximizes the above univariate function.\footnote{This optimization emerges precisely as the limit of the maximization problem in \cite{miyashita_ui_lqgd}, which yields an optimal symmetric signal structure in the finite homogeneous agent model.
 As elaborated in Section~\ref{sec_sym}, the current continuum population model exhibits a close connection between targeted disclosure and symmetric signal structures.}
We define the designer's value of the optimal targeted disclosure by
\begin{align*}
V_{\rm tg}^* \coloneqq \max_{m \in [0,1]} V_{\rm tg} (m).
\end{align*}
By definition, it holds that $V^* \ge V^*_{\rm tg}$.

We refer to targeted disclosure as \emph{partial disclosure} when it entails an interior point $m \in (0,1)$.
Additionally, the cases of corner points, $m = 0$ and $m = 1$, are named \emph{no disclosure} and \emph{full disclosure}, respectively.
The next lemma shows that the parameters $\alpha$ and $\beta$ defined as in \eqref{prob_tg}, coupled with $r$, are sufficient to identify the optimal targeted disclosure policy within the subclass.
Specifically, the optimum can be determined by the following parametric conditions, exhausting all possibilities; see Figure \ref{fig_tg1} for graphical illustration.

\begin{enumerate}[\rm (T1).]
\item $\alpha \le 0$ and $\alpha \le \beta$.
\item $\alpha > 0$ and $\beta > (\frac{1+r}{2}) \cdot \alpha$
\item $\alpha \ge \beta$ and $\beta \le (\frac{1+r}{2}) \cdot \alpha$.
\end{enumerate}

\begin{remark}
While (T2) is disjoint from other conditions, (T1) and (T3) are not from each other since both can hold if (and only if) $\alpha = \beta \le 0$.
In this knife-edge case, both full disclosure and no disclosure can achieve $V_{\rm tg}^*$, and this case is in fact the only occasion of having multiple maximizers of $V_{\rm tg}$.
\end{remark}

\begin{lemma} \label{lem_tg}
Within the class of targeted disclosure, the optimum is characterized as follows:
\begin{enumerate}[\rm i).]
\item No disclosure, $m^*=0$, is optimal if and only if {\rm (T1)} holds.
\item Partial disclosure, $m^* \in (0,1)$, is optimal if and only if {\rm (T2)} holds, where the optimal choice of $m^*$ is given by
\begin{align} \label{x_int}
m^* = \frac{\alpha}{2\beta - r \alpha}.
\end{align}
\item Full disclosure, $m^*=1$, is optimal if and only if {\rm (T3)} holds.
\end{enumerate}
Consequently, the value of the optimal targeted disclosure is given as follows:
\begin{align*}
V_{\rm tg}^* =\begin{cases}
0 &\text{if \quad {\rm (T1)} holds}, \\
\alpha^2/(4\beta - 4r\alpha) &\text{if \quad {\rm (T2)} holds}, \\
(\alpha-\beta)/(1-r)^2 &\text{if \quad {\rm (T3)} holds}.
\end{cases}
\end{align*}
\end{lemma}

The optimal targeted disclosure can be conveniently visualized by using diagrams, as depicted in Figure~\ref{fig_tg1}, where $\alpha$ and $\beta$ are taken on the horizontal and vertical axes, respectively.
The entire space is divided into three regions based on our parametric conditions.
No disclosure is optimal when $\alpha$ is small enough to be bounded by both $0$ and $\beta$.
Recalling that $\alpha = v+w$, this case can occur when the designer aims to minimize individual action volatility and correlation with the state
On the other hand, whenever no disclosure is suboptimal, full disclosure is optimal if $\beta$ is small enough for the point $(\alpha,\beta)$ to lie below the ray of slope $\frac{1-r}{2}$.
Recalling that $\beta = rw-u$, this case can occur when $u$ is large, meaning that the designer has an objective of inducing large correlations between different agents' actions.
The boundary of full and partial disclosure is determined by $r$, which reflects each agent's coordination motives.
As this parameter decreases, the regions are separated by a ray of smaller slope.

\begin{figure}[t]
\begin{center}
\includegraphics[width=15cm]{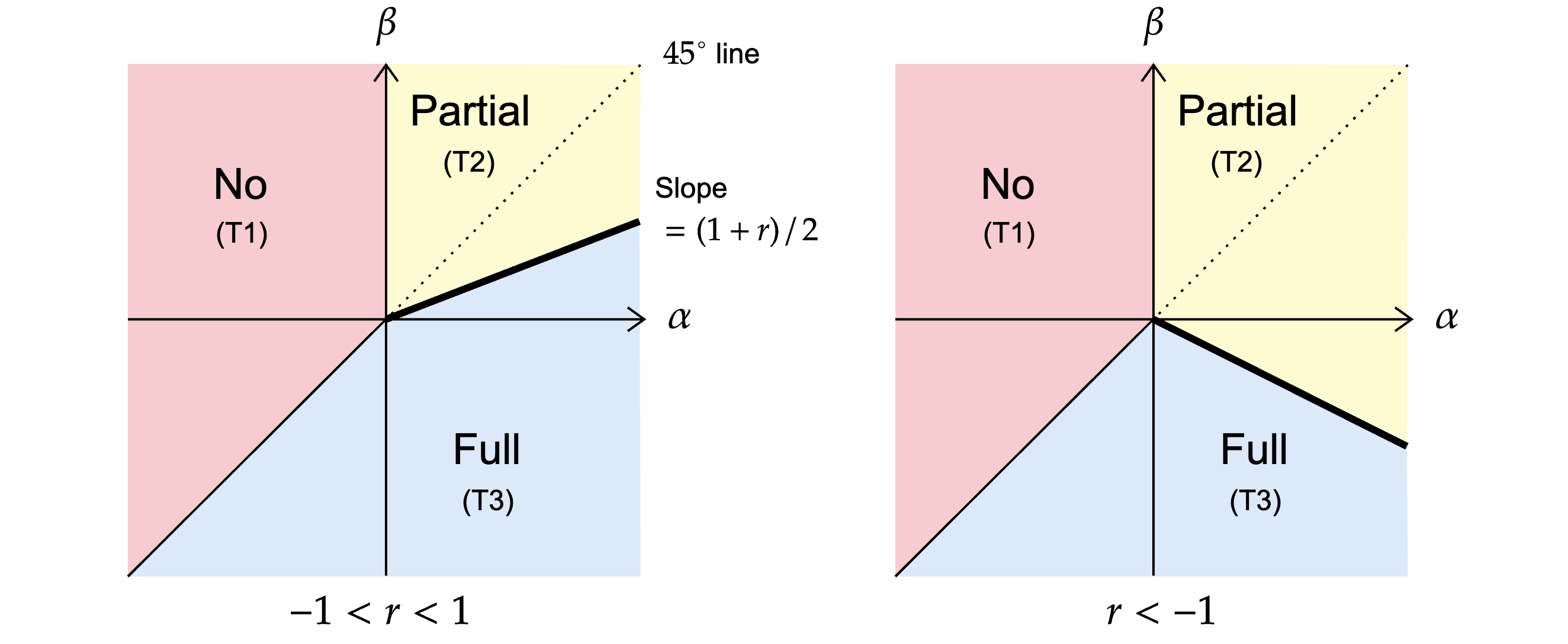}
\caption{The set of points $(\alpha,\beta)$ satisfying each of (T1), (T2), and (T3) are displayed. The left panel corresponds to the case of small strategic interactions, $r \in (-1,1)$, resulting in an upward-sloping ray that divides (T2) and (T3). In contrast, the right panel considers the case of large strategic substitutes, $r < -1$, resulting in a downward-sloping ray.}
\label{fig_tg1}
\end{center}
\end{figure}

Now, we are ready to present our main result in this section, which shows that the optimal targeted disclosure, given as in Lemma~\ref{lem_tg}, is actually globally optimal among all possible signal structures.

\begin{proposition} \label{prop_tg}
It holds that $V^* = V_{\rm tg}^*$.
Consequently, the optimal targeted disclosure in Lemma~\ref{lem_tg} is globally optimal.
\end{proposition}

The proof requires separate treatments of three cases, each corresponding to the parametric conditions considered in Lemma~\ref{lem_tg}.
By using properties of positive semidefinite kernels, we begin by providing an auxiliary result (Lemma~\ref{lem_bound}) that clarifies the joint implications of the two key constraints---obedience and positivity---in the form of inequalities imposed on the aggregates of $\xi$ and $\zeta$.
Then, given an arbitrary equilibrium moment $(\xi,\zeta)$, we show that $V_{\rm tg}^*$ serves as an upper bound on $V(\xi,\zeta)$ by identifying binding inequality constraints in each case.

To apply our findings, the next example considers the designer's problem of maximizing a convex combination of producer surplus and consumer surplus in Cournot games with linear demand functions as in \cite{vives1984, vives1999}.

\begin{figure}[t]
\begin{center}
\includegraphics[width=15cm]{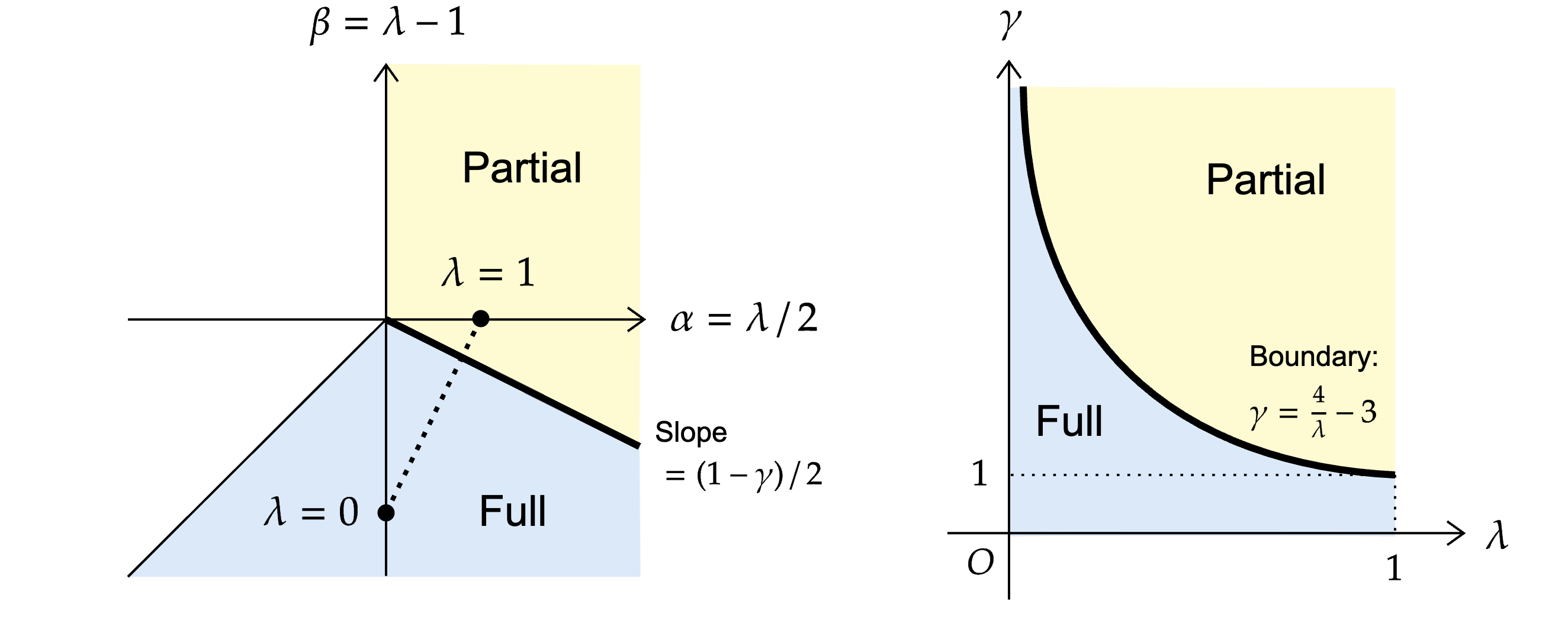}
\caption{In the left panel, the dashed line represents $(\alpha,\beta)$ in Example \ref{ex_market} for different values of $\lambda$. Full disclosure (resp.\ partial disclosure) is optimal if and only if $(\alpha,\beta)$ lies below (resp.\ above) the ray of slope $\frac{1-\gamma}{2}$. The right panel displays the region of $(\lambda,\gamma)$ for which each disclosure policy is optimal.}
\label{fig_tg2}
\end{center}
\end{figure}

\begin{example} \label{ex_market}
\cite{bm2013} and \cite{ui2015} explore a variant of Cournot games in \cite{vives1984, vives1999}, where each firm $i \in [0,1]$ supplies $a_i$ unit of a homogeneous product.
In this setup, the inverse demand function is a linear function of the total supply, given as $\theta - \gamma \int_0^1 a_j \d j$, where $\theta$ represents a random intercept and $\gamma > 0$ is a deterministic slope.
The cost to firm $i$ for producing $a_i$ units is specified as $a_i^2/2$.
These assumptions result in a quadratic and concave profit function for firm $i$, given as
\begin{align*}
\qty(\theta - \gamma \int_0^1 a_j \d j) \cdot a_i - \frac{a_i^2}{2}.
\end{align*}
As a consequence, the first-order condition yields the best-response strategy for firm $i$ as
\begin{align*}
a_i = \E_i [\theta] -\gamma \E_i \qty[\int_0^1 a_j \d j].
\end{align*}
This formulation aligns as a special case of the model in this section, setting $r = - \gamma$.

Now, producer surplus (PS) is defined as the aggregated profit of the firms, which can be calculated in expectation by using Proposition~\ref{prop_var} as follows:
\begin{align*}
{\rm PS} = - \gamma \int_0^1 \int_0^1 \Cov \qty[a_i,a_j] \d i \d j - \frac{1}{2} \int_0^1 \Var \qty[a_i] \d i + \int_0^1 \Cov \qty[a_i,\theta] \d i.
\end{align*}
Additionally, in the context of linear demand, consumer surplus (CS) can be defined as the square of the total supply, resulting in the following expectation:
\begin{align*}
{\rm CS} = \int_0^1 \int_0^1 \Cov \qty[a_i,a_j] \d i \d j.
\end{align*}
In their work, \cite{bm2013} identify an optimal signal structure that maximizes PS within the symmetric and Gaussian class of Example~\ref{ex_bm}.
Indeed, their signal structure can be certified as globally optimal using our Proposition~\ref{prop_tg}.
Furthermore, their analysis can be extended to consider the designer, who is concerned with both PS and CS.
Specifically, for $\lambda \in [0,1]$, let the designer's objective be given as  a convex combination of PS and CS,
\begin{align*}
V &= \lambda \cdot {\rm PS} + (1-\lambda) \cdot {\rm CS} \\
&= \underbrace{(1-\lambda-\lambda \gamma)}_{= \ u} \int_0^1 \int_0^1 \Cov \qty[a_i,a_j] \d i \d j
+ \underbrace{\qty(-\frac{\lambda}{2})}_{= \ v} \int_0^1 \Var \qty[a_i] \d i + \underbrace{\lambda}_{= \ w} \int_0^1 \Cov \qty[a_i,\theta] \d i,
\end{align*}
Notice that when $\lambda = 1/2$, maximizing $V$ is equivalent to maximizing total surplus, defined as the sum of PS and CS.

Given the above specification of $V$, we can calculate $\alpha = \lambda/2$ and $\beta = - (1-\lambda)$.
Then, by substituting these into Lemma~\ref{lem_tg}, we can determine an optimal targeted disclosure as a function of $\lambda$ and $\gamma$, which can be illustrated as in Figure~\ref{fig_tg2}.
Since $\alpha > \beta$, no disclosure can never be optimal.
In contrast, full disclosure is optimal if and only if $\beta \le (\frac{1+r}{2}) \cdot \alpha$, or equivalently,
\begin{align} \label{cond_market}
\gamma \le \frac{4}{\lambda} - 3.
\end{align}
Economically, this condition holds when the demand curve is steep enough to have a slope $\gamma$ exceeding a certain threshold, which is decreasing in $\lambda$, reflecting the weight on producer surplus (PS) in the designer's objective. 
In particular, if the designer only values CS (i.e., $\lambda = 0$), full disclosure is always optimal, regardless of the slope of the demand curve.
Alternatively, if the demand is insensitive to price (i.e., $\gamma \le 1$), full disclosure is always optimal, regardless of the designer's objective.
On the other hand, when \eqref{cond_market} is not satisfied, it is optimal to disclose information only to targeted firms of measure $m^* = \frac{\lambda}{\lambda \gamma-4(1-\lambda)}$.
Lastly, Proposition~\ref{prop_tg} ensures that these are globally optimal among all signal structures, including non-targeted, non-symmetric, and non-Gaussian ones.
\hfill $\triangle$
\end{example}

%%%%%%%%%%%%%%%%%%%%%%%%%%%%%%%%%%%%%%%%%%%%%
%%%%%%%%%%%%%%%%%%%%%%%%%%%%%%%%%%%%%%%%%%%%%

\subsection{Symmetric Disclosure}
\label{sec_sym}

It should be noted that targeted disclosure is not uniquely globally optimal, but there are other classes of signal structures that can achieve the designer's value as well.
One such alternative is a symmetric signal structure, which generates identically distributed signals for all agents.
In conjunction with the symmetry of the basic game, any symmetric signal structure induces a symmetric equilibrium moment, where $\zeta$ is constant, and $\xi$ is constant on the on-diagonal and off-diagonal subsets of $[0,1]^2$, respectively.
Formally, we define an equilibrium moment $(\xi,\zeta)$ as \emph{symmetric} if $\bar{\xi}_1 \equiv \xi(t,t)$, $\bar{\xi}_2 \equiv \xi(s,t)$, and $\bar{\zeta} \equiv \zeta(t)$ hold for all $s,t \in [0,1]$ with $s \neq t$.
For a symmetric equilibrium moment, its value is denoted as $V(\xi,\zeta) = V_{\rm sym} (\bar{\xi}_1,\bar{\xi}_2,\bar{\zeta})$.
We denote $V_{\rm sym}^*$ to express the maximum value across all symmetric equilibrium moments.

Due to the BNE-BCE equivalence result of \cite{bm2013}, the designer can implement any symmetric equilibrium moment by using a symmetric signal structure.
The next proposition explicitly constructs such a signal structure for a special kind of symmetric equilibrium moment, indexed by $m \in [0,1]$.
Then, we show that the value of the corresponding symmetric equilibrium moment coincides exactly with $V_{\rm tg} (m)$, the value attained by the targeted disclosure entailing the informed agents of measure $m$.
By means of this observation, and by Proposition~\ref{prop_tg}, we can conclude that there always exists a globally optimal symmetric signal structure.

\begin{proposition} \label{prop_sym}
Given any $m \in [0,1]$, we let
\begin{align*}
\bar{\xi}_1= \frac{m}{(1-rm)^2}, \quad \bar{\xi}_2 = \frac{m^2}{(1-rm)^2}, \quad \bar{\zeta}  = \frac{m}{1-rm}.
\end{align*}
Then, $(\bar{\xi}_1,\bar{\xi}_2,\bar{\zeta})$ is a symmetric equilibrium moment such that $V_{\rm sym}(\bar{\xi}_1,\bar{\xi}_2,\bar{\zeta}) = V_{\rm tg}(m)$.
Consequently, $V^*_{\rm sym} = V^*_{\rm tg} = V^*$ holds.
In particular, if $\theta \sim \calN(0,1)$, then $(\bar{\xi}_1,\bar{\xi}_2,\bar{\zeta})$ can be induced by generating each agent's signal as follows:
\begin{align*}
x(t) = \qty(\frac{m}{1-rm}) \cdot \theta + \qty(\frac{\sqrt{m(1-m)}}{1-rm}) \cdot \epsilon(t) \quad \text{where} \quad \epsilon(t) \sim_{\rm i.i.d.} \calN(0,1),
\end{align*}
where $\epsilon(t)$ is independently distributed from $\theta$.
\end{proposition}

\begin{remark}
When the designer is oriented to welfare maximization, one can confirm that the optimal symmetric signal structure in this result induces the welfare-maximizing Bayes correlated equilibrium in Corollary 13 of \cite{ui2015}.
\end{remark}

While both targeted and symmetric signal structures can achieve the same objective, each may offer distinct practical advantages.
Targeted disclosure, for instance, may be favored for its simplicity, characterized by its bang-bang property wherein each agent either observes the entire state realization or nothing at all.
Implementing targeted disclosure requires only the designer to determine the set of informed agents, without the need for additional noise.
In contrast, an equivalent symmetric signal structure necessitates idiosyncratic randomization.

On the other hand, symmetric disclosure holds greater robustness against information spillover among agents, which may occur after private signals are transmitted from the designer.
To illustrate, consider a scenario where agents can communicate with a few of their neighbors to share their signal realizations before choosing actions.
Under targeted disclosure, even if an agent is uninformed and did not receive disclosure from the designer, she can still eliminate all uncertainty if she has the opportunity to communicate with any single informed neighbor.
This spillover issue is practically relevant because the resulting economic outcome after information sharing can differ significantly from what the designer intended to achieve.
In contrast, under the symmetric disclosure as in Proposition~\ref{prop_sym}, the agent can still reduce uncertainties, but complete elimination requires collecting an infinite number of signal realizations from others.

One might argue that fairness is a virtue of symmetrized information: Unlike targeted disclosure, symmetric disclosure treats all agents equally, at least in the ex-ante sense.
However, targeted disclosure can also achieve the same goal if we allow for randomization over informed agents.

To illustrate this idea, consider the following \emph{randomized targeted disclosure}:
First, divide the population $[0,1]$ into $n$ subsets of equal size as follows:
\begin{align*}
M_1 = \left[0,\frac{1}{n} \right),\, \ldots,\, M_{n-1} = \left[\frac{n-2}{n},\frac{n-1}{n} \right),\, M_n = \left[\frac{n-1}{n}, 1 \right].
\end{align*}
Second, randomly select $k$ distinct subsets from the above with equal likelihoods.
Third, fully disclose the state realization to all agents in the union of the selected subsets, while the remaining agents receive no information.
Notice that the measure of the informed agents remains constant at $k/n$, regardless of the subsets chosen.
Therefore, a unique equilibrium under randomized targeted disclosure is essentially the same as that under the targeted disclosure, which sets $m = k/n$.
Moreover, since each subset is selected equally likely, the designer's expected objective remains the same as $V_{\rm tg}(k/m)$.
Hence, by choosing natural numbers $(n,k)$ to make $k/n$ arbitrarily close to the optimal measure $m^*$, randomized targeted disclosure can attain the designer's value while ensuring that every agent has an equal chance of receiving information.
The procedure employs randomization solely for the purpose of selecting informed agents: Implementing this policy simply requires shuffling $n$ cards and flipping $k$ from the top.

%As one can see from the above construction, the randomized targeted disclosure can be regarded as a convex combination of a bunch of targeted disclosure. In \cite{miyashita_ui_lqgd}, we delve into the analysis of information design problems in LQG environments 

\subsection{Public Disclosure}

In practice, the designer is often restricted to using only public disclosure, resulting in an signal realization that is commonly observed by all agents.
For instance, in the scenario of the beauty contest model of \cite{morris2002shin}, the designer is interpreted as a central bank, who is required to adhere to the principle of central bank transparency through committing to public disclosure.\footnote{Public disclosure is considered in several contexts of information design with multiple receivers, e.g., \cite{alonso2016camara}, \cite{laclau2017renou}, and \cite{kosterina2022}. As mentioned in Section~4.1 of the survey paper by \cite{kamenica2019}, the analysis of public signals can be reduced to a single-agent Bayesian persuasion problem by interpreting an equilibrium action profile as the representative agent's action. In the finite population LQG model, \cite{miyashita_ui_lqgd} characterized an optimal public signal structure by applying the solution method of \cite{tamura2018} for a quadratic Bayesian persuasion problem.}

Formally, we say that a signal structure is \emph{public} when agent's signal $x(t) \equiv x$ is given as an identical random variable for every $t \in [0,1]$.
We denote by $V_{\rm pub}^*$ the highest designer's objective attainable through public signal structures.
These signal structures are in stark contrast to targeted disclosure, which is featured by its selective nature, as well as to the signal structure in Proposition \ref{prop_sym}, featured by its idiosyncratic randomization.
Since targeted disclosure can achieve global optimality, evaluating the gap between $V_{\rm tg}$ and $V_{\rm pub}$ allows us to gauge the loss incurred by the designer when restricting herself to communicating with the  population only through public disclosure.

Given any public signal $x$, it is not hard to show that a unique equilibrium under the public signal structure is given as $f(t) = \frac{\E \qty[\theta \mid x]}{1-r}$ for all $t \in [0,1]$, from which the following second moments are induced:
\begin{align*}
\Cov \qty[f(s),f(t)] = \frac{z}{(1-r)^2}
\quad \text{and} \quad
\Cov \qty[f(t),\theta] = \frac{z}{1-r},
\end{align*}
where $z = \E \qty[\Var \qty[\theta \mid \eta]]$ represents the residual state variance that remains after observing the public signal $x$.
Notice that $z$ ranges between $0$ and $1$ by construction.
Consequently, $V_{\rm pub}^*$ can be identified as the value of the following maximization:
\begin{align*}
V_{\rm pub}^* = \max_{z \in [0,1]} \quad (v+w) \cdot \frac{z}{(1-r)^2} + c \cdot \frac{z}{1-r}.
\end{align*}
Since the objective an be rewritten as $\frac{(\alpha - \beta)z}{(1-r)^2}$, we see that the problem admits a unique solution $z^* = 0$ if $\alpha < \beta$, while $z^* = 1$ if $\alpha > \beta$, corresponding to the case of no disclosure and full disclosure, respectively.
Hence, since the problem admits no interior solution (unless $\alpha = \beta$), we have $V_{\rm pub}^* < V^*$ whenever (T2) is satisfied.
This suggests that the designer incurs a strict loss by limiting herself to public disclosure whenever partial disclosure is globally optimal.

%%%%%%%%%%%%%%%%%%%%%%%%%%%%%%%%%%%%%%%%%%%%%
%%%%%%%%%%%%%%%%%%%%%%%%%%%%%%%%%%%%%%%%%%%%%

\section{Concluding Comments} \label{sec_conc}

%We conclude this paper by commenting on its connections to some existing works, ranging among multiple brunches of the economics literature.
%\subsection{Integration of Random Variables} \label{sec_conc_pettis}
%\subsection{Equilibrium Analysis in Large Population Games} \label{sec_conc_eq}

This paper contributes to the economics literature in various aspects.
First, by advocating the integral notion of \cite{pettis1938}, we offer a solid mathematical foundation for large population games with incomplete information.
Second, we provide a parsimonious sufficient condition on payoff structures that guarantee the uniqueness of an equilibrium in a broad class of games, including large scale network games and LQG games.
Third, we conduct comprehensive analysis on information design problems in the symmetric LQG environment and establish the global optimality of practical information disclosure policies.

Several papers on large population games, such as \cite{bm2013}, can be read without prior knowledge of Pettis integral, and their economic implications remain valid even when integration is interpreted naively.
Therefore, the first contribution of this paper lies in providing mathematical formality for the economic insights derived from earlier works, rather than in identifying errors or incorrect conclusions, by employing the Pettis-integral approach.
We believe that this approach is among the simplest and least technically demanding ones. %as the results in Section \ref{sec_pettis} are proven using standard tools in economic theory.
In addition, the implications of our findings, demonstrated through the running example, may align with economists' intuitions.
In light of these, we propose that this approach can serve as a useful instrument for future studies of large population games, offering a means to avoid tricky issues related to the aggregation of a continuum of random variables.

As cautioned in \cite{khan1999sun}, the Pettis-integral approach can ``avoid'' the measurability issue, but it does not ``resolve'' the issue and leaves some interpretation challenges.\footnote{The lack of ex-post interpretations of Pettis integral is well recognized and discussed in \cite{al1995} and \cite{uhlig1996}, who firstly advocate for this approach as modeling tools of large economies.}
To illustrate their discussion, consider a collection of random variables $\{f(t)\}_{t \in T}$ that are Pettis-integrated to the aggregated random variable $\bar{f} \equiv \int f(t) \d \nu(t)$.
While $\bar{f}$ has an ex-post realization $\bar{f}(\omega)$ for each $\omega$ in the underlying probability space, it may lack meaningful relations to the realized sample path of the process, i.e., $\{f(t,\omega)\}_{t \in T}$, since a typical sample path need not be Lebesgue integrable across $t$.
In the context of Example~\ref{ex_bm}, the Pettis integral approach allows us to define the aggregated action $A$ from an ex-ante standpoint, but it does not reveal how the realization of $A$ can be related to the realizations of the state, signals, and individual actions from an ex-post standpoint.
Moreover, we can not give probabilistic interpretations to Proposition \ref{prop_var} in a way similar to the traditional LLN because the convergence of the mean of the process is not stated in a realized state-wise manner.
Rather, the proposition should be understood as calculating the moment of the aggregate from the ex-ante standpoint, prior to the resolution of uncertainty about $\omega$.\footnote{\cite{sun2006} proposes a different approach to integrate a continuum of random variables by considering an extension of the measurable space $(T \times \Omega, \Sigma \otimes \Pi, \nu \otimes \P)$, on which the real-valued function $f:T \times \Omega \to \R$ acts measurably. On this extended measurable space, he establishes the ``exact'' LLN, stating that the sample mean $\int_T f(t,\omega) \d \nu(t)$ is equal to the ex-ante mean of the process $\P$-almost surely. He defines the extension by appealing to the Fubini property, the interchange of the integrals over $T$ and $\Omega$, which bears similar to the idea of Pettis integral.}

\begin{comment}
\begin{quote}
{\it Either the ensemble of risk facing a decision maker consists solely of collective risk, in which case $\mu(f)$ {\rm [the distribution of the aggregate]} has meaning and can be an argument in the payoff function; or the assumption is not justified, in which case the measurability problem remains important and the distribution of the sample function cannot be given meaning, much less shown to filter out the particular state of nature.}
\end{quote}
\end{comment}

Under what circumstances would the Pettis integral approach be appropriate for modeling large scale strategic situations with incomplete information?
In our perspective, the combination of two ingredients---aggregative and Bayesian elements---is crucial.
In the running example, the aggregated action is the sole strategic argument in agents' payoff functions that matter their incentives.
The notion of Pettis integral relies on the interchange of integration and probabilistic operations.
Applying this to the best-response formula \eqref{br_bm}, we observe a specific way of how Bayesian agents evaluate strategic uncertainty:
The aggregated action is best estimated as the total sum of the individual uncertainties associated with each opponent's strategic decision.
If one agrees that viewing aggregated uncertainty as the sum of individual uncertainties guides ex-ante payoff maximization, then the Pettis integral approach becomes appealing, as the notion is essentially defined by embracing this idea as an axiom.
On the other hand, the Pettis integral approach may lack clear economic interpretations when the modeler's interest lies in ex-post realizations of the aggregate or sample path.\footnote{\cite{khan1999sun} reach a similar conclusion to us; see the last paragraph of their Section~7.}

Applying the Pettis-integral approach to games requires some regularity conditions on primitives, such as the measurability of payoff structures. 
More crucially, we may need to confine attention to strategy profiles holding certain measurability properties in order to properly define the aggregate.
As pointed out by \cite{al2008}, since measurability implicitly assume some degree of similarity among agents, focusing on such strategies could be an undesirable restriction, as an equilibrium is typically perceived as a consequence of decentralized strategic decisions.
On the other hand, it might be reasonable to predict that agents tend to behave similarly when they share similar primitives, such as payoff functions, private information, or sets of neighbors.
In light of this, a result on equilibrium existence can be seen as a validation of such predictions.\footnote{While no existence result is available under as general assumptions as those of our uniqueness result, existence can often be established through direct approaches by exploiting the structure of a game; for instance, the matching coefficient method is applicable when the game exhibits symmetry.
Also, as noted earlier in Section~\ref{sec_eq_unique}, there are several known results on existence within the literature that apply under particular model specifications.}

%as it shows that by imposing certain measurability conditions on primitives, there exists an equilibrium wherein agents' strategies are measurable in the same degree.

%\subsection{Information Design with Multiple Receivers} \label{sec_conc_info}

The global optimality of targeted disclosure in our analysis partially relies on the non-atomicity of the population, which allows the designer to finely adjust the optimal measure of targeted agents and select it as a continuous variable from the interval $[0,1]$.
\cite{smolin2023yamashita} also consider targeted disclosure in a finite population setting and establish its optimality in two applications.
In their second application concerning belief polarization, however, achieving the exact optimum requires an even number of agents due to the fact that the choice of targeted disclosure is discrete for finite agents.
Nevertheless, given the continuity of the designer's objective function, targeted disclosure can still achieve the global optimum approximately  when the number of agents is finite but sufficiently large.

%%%%%%%%%%%%%%%%%%%%%%%%%%%%%%%%%%%%%%%%%%%%%
%%%%%%%%%%%%%%%%%%%%%%%%%%%%%%%%%%%%%%%%%%%%%

%\newpage
\appendix

\begin{center}
\Large{{\bf Appendix}}
\end{center}

%%%%%%%%%%%%%%%%%%%%%%%%%%%%%%%%%%%%%%%%%%%%%
%%%%%%%%%%%%%%%%%%%%%%%%%%%%%%%%%%%%%%%%%%%%%

\section{Mathematical Preliminaries}
\label{app_kernel}

In this appendix, we invoke the mathematical notion of kernels and present some results on it.
These preparations are instrumental for proving our results in Section \ref{sec_eq} and Section \ref{sec_info}.

Throughout, let $(T,\Sigma,\nu)$ be a finite measure space.
Denote by $\calL_2(\nu)$ be the usual Hilbert space that collect all (equivalence classes of) square-integrable functions from $T$ to $\R$, equipped with the inner product,
\begin{align*}
\langle \phi,\psi \rangle_{\calL_2(\nu)} \coloneqq \int \phi(t) \psi(t) \d \nu(t), \quad \forall \phi, \psi \in \calL_2(\nu).
\end{align*}
Let $\|\phi\|_{\calL_2(\nu)} = \langle f,f \rangle_{\calL_2(\nu)}^{1/2}$ be the induced norm.
Analogously, let $\calL_2(\nu \otimes \nu)$ be the Hilbert space of square-integrable bivariate functions from $T^2$ to $\R$.

\begin{definition} \label{def_kernel}
A jointly measurable function $K:T^2 \to \R$ is called a \emph{(measurable) kernel}.
In addition, we say that:
\begin{itemize}
\item  $K$ is \emph{square integrable} if
\begin{align*}
\|K\|_{\calL_2(\nu \otimes \nu)} \coloneqq \qty(\int \int |K(t,t')|^2 \d \nu(t) \d \nu(t'))^{\frac{1}{2}} < \infty.
\end{align*}
\item  $K$ is \emph{integrable on the diagonal} if
\begin{align*}
\|K\|_{\calL_2(\nu)} \coloneqq \int \int |K(t,t)| \d \nu(t) < \infty.
\end{align*}
\item  $K$ is \emph{(essentially) bounded on the diagonal} if
\begin{align*}
\|K\|_{\calL_\infty(\nu)} \coloneqq \esup_{t \in T} |K(t,t)| < \infty.
\end{align*}
\end{itemize}
\end{definition}

As mentioned in \eqref{def_R} in the main text, any square-integrable kernel $K$ naturally defines an integral operator $\mathbf{K}: \calL_2(\nu) \to \calL_2(\nu)$.
Then, $\lambda \in \R$ is called an \emph{eigenvalue} of $\mathbf{K}$ when $\mathbf{K}\phi = \lambda \phi$ holds for some \emph{eigenvector} $\phi \in \calL_2(\nu) \setminus \{0\}$, where $0$ denotes the zero element of $\calL_2(\nu)$, and denote by $\Lambda(\mathbf{K})$ the set of eigenvalues.
Additionally, we define the \emph{numerical range} of $\mathbf{K}$ by
\begin{align*}
\Phi(\mathbf{K}) \coloneqq \qty{\frac{\langle \phi,\mathbf{K}\phi \rangle_{\calL_2(\nu)}}{\|\phi\|_{\calL_2(\nu)}}: \phi \in \calL_2(\nu) \setminus \{0\}},
\end{align*}
where each value in $\Phi(\mathbf{K})$ is called the \emph{Rayleigh quotient} of $\phi$.
Moreover, the \emph{operator norm} of $\mathbf{K}$ is defined by
\begin{align*}
\|\mathbf{K}\|_{\rm op} \coloneqq \sup_{\phi \in \calL_2(\nu) \setminus \{0\}} \frac{\|\mathbf{K} \phi\|_{\calL_2(\nu)}}{\|\phi\|_{\calL_2(\nu)}},
\end{align*}
which is finite given that $K$ is square integrable.
These variables can be related to one another as in the next lemma.

\begin{lemma} \label{lem_num}
If $K:T^2 \to \R$ is a square-integrable kernel, then
\begin{align*}
\co \qty(\Lambda(\mathbf{K})) \subseteq \Phi(\mathbf{K}) \subseteq \bigl[- \|\mathbf{K}\|_{\rm op}, \|\mathbf{K}\|_{\rm op} \bigr],
\end{align*}
where $\co (\cdot)$ denotes the convex full of a set.
In addition, if $K$ is symmetric, or in the game-theoretic context, undirected, then $\co \qty(\Lambda(\mathbf{K})) = \Phi(\mathbf{K})$.
\end{lemma}

\begin{proof}
By definition, one can see that $\Lambda(\mathbf{K}) \subseteq \Phi(\mathbf{K})$.
Then, since $\Phi(\mathbf{K})$ is convex by Hausdorff--Toeplitz theorem \citep[Proposition 11.33 of][]{brezis}, it follows that $\co \qty(\Lambda(\mathbf{K})) \subseteq \Phi(\mathbf{K})$.
Moreover, $\Phi(\mathbf{K}) \subseteq \qty[- \|\mathbf{K}\|_{\rm op}, \|\mathbf{K}\|_{\rm op} ]$ follows from the Cauchy--Schwartz inequality.
Lastly, if $K$ is symmetric, then $\mathbf{K}$ is self-adjoint, and thus, the Fischer--Courant theorem \citep[Theorem 4 in pp.\ 318--319 of][]{lax} implies that $\inf \Lambda (\mathbf{K}) = \inf \Phi(\mathbf{K})$ and $\sup \Lambda(\mathbf{K}) = \sup \Phi(\mathbf{K})$, from which we conclude that $\co \qty(\Lambda(\mathbf{K})) = \Phi(\mathbf{K})$.
\end{proof}

\begin{definition} \label{def_psd}
A kernel $K:T^2 \to \R$ is \emph{symmetric} if for any $s,t \in T$, $K(s,t) = K(t,s)$.
Moreover, $K$ is \emph{positive semidefinite} if for any $n \in \N$ and $t_1,\ldots,t_n \in T$,
\begin{align} \label{gram}
\mqty[K(t_1,t_1) & \cdots & K(t_1,t_n) \\ \vdots & \ddots & \vdots \\ K(t_n,t_1) & \cdots & K(t_n,t_n)] \succeq O.
\end{align}
We say that $K$ is negative semidefinite if $(-K)$ is positive semidefinite.
\end{definition}

\begin{remark}
There is an alternative way of defining positive semidefinite kernels through reproducing kernel Hilbert spaces (Definition \ref{def_rkhs}), but it turns out to be equivalent to the above one; see Chapter 4 of \cite{steinwart2008christmann} for details.
Also, it may be common in the machine learning literature to include positive semidefiniteness in the definition of kernels, but we do not do so in order to regard a payoff structure as a kernel.
\end{remark}

Let us mention a few basic properties of positive semidefinite kernels.
First, from the case of $n=1$ in \eqref{gram}, we have $K(t,t) \ge 0$ for all $t \in T$.
Second, from the case of $n=2$, we obtain the following \emph{Cauchy--Schwartz inequality}:
\begin{align} \label{k_cs}
K(s,t) \le \sqrt{K(s,s)} \sqrt{K(t,t)}, \quad \forall s,t \in T.
\end{align}
By this property, we immediately see that $\|K\|_{\calL_2(\nu \otimes \nu)}\le \|K\|_{\calL_2(\nu)} \le \|K\|_{\calL_\infty(\nu)}$, suggesting that any positive semidefinite kernel is square integrable whenever it is integrable or bounded on the diagonal.
Third, if $K$ is square integrable, then it holds that
\begin{align} \label{h_psd}
0 \le \int \int \phi(s) K(s,t) \phi(t) \d \nu(s) \d \nu(t) \le \|K\|_{\calL_2(\nu \otimes \nu)} \cdot \|\phi\|^2_{\calL_2(\nu)}, \quad \forall \phi \in \calL_2(\nu),
\end{align}
which follows from the classical spectral theorem; see Lemma \ref{lem_eigen_pos} below.
In particular, this implies that the integral $\int \int K(s,t) \d\nu(s) \d\nu(t)$ is non-negative, as we can let $\phi$ be a constant function in \eqref{h_psd}.

Now, let $K$ be any square-integrable kernel.
By using the Cauchy--Schwartz inequality on $\calL_2(\nu)$, one can easily see that $\|\mathbf{K}\phi\|_{\calL_2(\nu)} < \|K\|_{\calL_2(\nu \otimes \nu)} \cdot \|\phi\|_{\calL_2(\nu)}$, suggesting that $\mathbf{K}$ is a bounded linear operator, acting from $\calL_2(\nu)$ to itself, with the operator norm bounded by $\|K\|_{\calL_2(\nu \otimes \nu)}$.
In fact, $\mathbf{K}$ satisfies compactness, a property stronger than boundedness \citep[Theorem 4 in p.\ 247 of][]{lax}, which in turn implies that $\mathbf{K}$ has at most countably many eigenvalues, including geometric multiplicity \citep[Theorem 6 in p.\ 238 of][]{lax}.

\begin{lemma}
\label{lem_eigen_com}
Let $K:T^2 \to \R$ be a square-integrable kernel.
Then, $\mathbf{K}$ is a compact linear operator on $\calL_2(\nu)$ with $\|\mathbf{K}\|_{\rm op} \le \|K\|_{\calL_2(\nu \otimes \nu)}$.
Consequently, $\mathbf{K}$ has at most a countable set of eigenvalues $\{\lambda_i\}_{i \in I}$ with no accumulation point except, possibly, $0$.
\end{lemma}

If $K$ is symmetric, then $\mathbf{K}$ is self-adjoint, in which case all eigenvalues of $\mathbf{K}$ are real numbers.
Moreover, assuming that $K$ is positive semidefinite, all eigenvalues of $\mathbf{K}$ are non-negative, and the following spectral theorem holds true \citep[Theorem A.1 of][]{steinwart2012scovel}.

\begin{lemma}
\label{lem_eigen_pos}
Let $K:T^2 \to \R$ be a square-integrable and positive semidefinite kernel.
Then, all eigenvalues of $\mathbf{K}$ are weakly positive and can be aligned in decreasing order, $\lambda_1 \ge \lambda_2 \ge \cdots \ge 0$.
Moreover, there exists an orthonormal system $\{\psi_i\}_{i \in I}$ in $\calL_2(\nu)$, where each $\psi_i$ is the eigenvector of $\mathbf{K}$ associated with $\lambda_i$, such that
\begin{align*}
\mathbf{K}\phi = \sum_{i \in I} \lambda_i \langle \phi, \psi_i \rangle_{\calL_2(\nu)} \psi_i, \quad \forall \phi \in \calL_2(\nu).
\end{align*}
\end{lemma}

\begin{remark}
In what follows, we assume without loss of generality that $I$ is an infinite set by considering $\lambda_i = 0$ whenever $I$ is finite and $i > |I|$.
\end{remark}

When $K$ is a continuous function on a compact topological space $T^2$, the above spectral representation can be further refined.
According to Mercer's theorem \citep[Theorem 11 in p.\ 343 of][]{lax}, we can choose continuous functions $\{e_i\}_{i \in I}$, as the representatives of the eigenvectors $\{\psi_i\}_{i \in I}$, in such a way that the continuous kernel $K$ enjoys the representation,
\begin{align} \label{mercer_rep}
K(s,t) = \sum_{i \in I} \lambda_i e_i(s) e_i(t), \quad \forall s,t \in T,
\end{align}
where the series converges absolutely and uniformly.
Though Mercer's theorem is not directly applicable when $K$ is not continuous, a recent result by \cite{steinwart2012scovel} establishes a weaker form Mercer's representation, showing \eqref{mercer_rep} holds for almost every $(s,t)$, that applies to measurable kernels.

We need a few more notions in order to state and apply their results for our proof.
In stating the next definition, we should be aware of the difference between a square-integrable function, which is an element of $H$, and an element of $\calL_2(\nu)$, i.e., the latter is an equivalent class of the former.

\begin{definition} \label{def_rkhs}
Let $H \subseteq \R^T$ be a Hilbert space, consisting of functions from $T$ to $\R$.
The space $H$ is called a \emph{reproducing kernel Hilbert space} (RKHS) if for every $t \in T$, the evaluation mapping $h \mapsto h(t)$ is a continuous function from $H$ to $\R$.
Moreover, a positive semidefinite kernel $K:T^2 \to \R$ is called a \emph{reproducing kernel} of $H$ if $K(\cdot,t) \in H$ for every $t \in T$, and if the following \emph{reproducing property} holds:
\begin{align*}
h(t) = \langle h, K(\cdot,t) \rangle_H, \quad \forall h \in H, \quad \forall t \in T.
\end{align*}
\end{definition}

Known as the Moore--Aronszajn theorem, there exists a one-to-one relationship between positive semidefinite kernels and RKHSs.
Specifically, in \cite{steinwart2008christmann}, Theorem 4.20 shows that every RKHS has a unique reproducing kernel.
Conversely, their Theorem 4.21 shows that for every positive semidefinite kernel $K$, there exists a unique RKHS $H$ for which $K$ serves as the reproducing kernel.
By virtue of this, we refer to $H$ as the RKHS of $K$.

\begin{lemma} \label{lem_mercer_diag}
Let $K$ be a positive semidefinite kernel that is integrable on the diagonal, $H$ be the RKHS of $K$, and $\{\lambda_i, \psi_i\}_{i \in I}$ be the eigenpairs of $\mathbf{K}$ given as in Lemma \ref{lem_eigen_pos}.
Then, there exist the representatives $\{e_i\}_{i \in I}$ of $\{\psi_i\}_{i \in I}$ such that $\{\sqrt{\lambda_i} e_i\}_{i \in I}$ is an orthonormal system in $H$.
Moreover, it holds that
\begin{align} \label{mercer_diag}
\sum_{i \in I} \lambda_i e_i(t)^2 \le K(t,t)
\quad \text{and} \quad
\sum_{i \in I} \lambda_i \qty|e_i(s)| \qty|e_i(t)|  \le \sqrt{K(s,s)} \sqrt{K(t,t)}.
\quad \forall s,t \in T.
\end{align}
\end{lemma}

\begin{proof}
We only need to prove the inequalities in \eqref{mercer_diag} since the rest is implies by Lemma 2.3 and Lemma 2.12 of of \cite{steinwart2012scovel}.
For the first inequality, since every Hilbert space has an orthonormal basis \citep[Theorem 9 in p.\ 60 of][]{lax}, we can find a (possibly, uncountable) family $\{\tilde{e}_j\}_{j \in J}$ in $H$ such that the union $\{\sqrt{\lambda_i}e_i\}_{i \in I} \cup \{\tilde{e}_j\}_{j \in J}$ is an orthonormal basis in $H$.
Then, since $K(\cdot,t) \in H$ for all $t \in T$ by Definition \ref{def_rkhs}, it follows from Lemma 8 in p.\ 60 of \cite{lax} that $K(\cdot,t)$ is represented as
\begin{align*}
K(\cdot,t) = \sum_{i \in I} \alpha_i(t) \sqrt{\lambda_i} e_i + \sum_{j \in J} \beta_j(t) \tilde{e}_j,
\end{align*}
where $\alpha_i(t) = \langle K(\cdot,t), \sqrt{\lambda_i} e_i \rangle_H$, $\beta_j(t) = \langle K(\cdot,t), \tilde{e}_j \rangle_H$, and the series converges absolutely with respect to $\|\cdot\|_H$. 
Moreover, by the reproducing property of $K$, we have
\begin{align*}
\alpha_i(t) = \bigl\langle K(\cdot,t), \sqrt{\lambda_i} e_i \bigr\rangle_H = \sqrt{\lambda_i} e_i(t) \quad \text{and} \quad
\beta_j = \bigl\langle K(\cdot,t), \tilde{e}_j \bigr\rangle_H = \tilde{e}_j(t),
\end{align*}
from which
\begin{align*}
K(\cdot,t) = \sum_{i \in I} \lambda_i e_i(t) e_i + \sum_{j \in J} \tilde{e}_j(t) \tilde{e}_j.
\end{align*}
Evaluating this expression at $t$, we have $K(t,t) \ge \sum_{i \in I} \lambda_i e_i(t)^2$, as $\tilde{e}_j(t)^2 \ge 0$ for all $j \in J$.

For the second inequality, let $n \in \N$ be any finite number.
Using the Cauchy--Schwartz inequality on $\R^n$, our previous conclusion implies that
\begin{align*}
\sum_{i=1}^n \lambda_i \qty|e_i(s)| \qty|e_i(t)|
%\sum_{i=1}^n \sqrt{\lambda_i}|e_i(s)| \cdot \sqrt{\lambda_i}|e_i(t)| 
&\le \qty(\sum_{i=1}^n \lambda_i e_i(s)^2)^{\frac{1}{2}} \cdot \qty(\sum_{i=1}^n \lambda_i e_i(t)^2)^{\frac{1}{2}} \\
&\le \qty(\sum_{i \in I} \lambda_i e_i(s)^2)^{\frac{1}{2}} \cdot \qty(\sum_{i \in I} \lambda_i e_i(t)^2)^{\frac{1}{2}}
\le \sqrt{K(s,s)} \sqrt{K(t,t)}.
\end{align*}
Since $I$ is at most countable, letting $n \to \infty$ yields the desired result.
\end{proof}

The next lemma is weak Mercer's theorem of \cite{steinwart2012scovel} that is stated by putting their Lemma 2.3, Lemma 2.12, and Corollary 3.2 altogether.

\begin{lemma}
\label{lem_mercer}
Let $K$ be a positive semidefinite kernel that is integrable on the diagonal, and let $\{\lambda_i\}_{i \in I}$ and $\{e_i\}_{i \in I}$ be as in Lemma \ref{lem_mercer_diag}.
Then, for each $t \in T$, there exists a $\nu$-null set $N_t \subseteq T$ such that
\begin{align*}
K(t,t') = \sum_{i \in I} \lambda_i e_i(t) e_i(t'), \quad \forall t' \in T \setminus N_t,
\end{align*}
where the series converges absolutely.
In particular, the representation holds for $\nu \otimes \nu$-almost every $(t,t') \in T^2$.
\end{lemma}

Now, leveraging the preliminary results introduced so far, we derive an upper bound on the eigenvalues of the Hadamard product of two integral operators.
The next lemma plays a pivotal role in the proof of Proposition \ref{prop_unique}.
Mathematically, it can be seen as an infinite-dimensional analog of Schur's product theorem for positive semidefinite matrixes \citep[Theorem 3.4 of][]{styan1973}.

\begin{lemma} \label{lem_eigen}
Let $K:T^2 \to \R$ be a positive semidefinite kernel that is bounded on the diagonal.
Also, let $R:T^2 \to \R$ be a square-integrable kernel satisfying {\rm (R1)}.
Define the Hadamard product $\mathbf{K}\circ \mathbf{R}:\calL_2(\nu) \to \calL_2(\nu)$ by
\begin{align*}
\qty[\mathbf{K}\circ \mathbf{R}](\phi) \coloneqq \int \qty[k(\cdot,t)r(\cdot,t)]\phi(t) \d \nu (t), \quad \forall \phi \in \calL_2(\nu).
\end{align*}
Then, $\mathbf{K}\circ \mathbf{R}$ is a compact linear operator, and all real eigenvalues of $\mathbf{K}\circ \mathbf{R}$ are less than $\|K\|_{\calL_\infty(\nu)}$.
In particular, if $\mathbf{K}$ has at least one strictly positive eigenvalue, then all real eigenvalues of $\mathbf{K}\circ \mathbf{R}$ are strictly less than $\|K\|_{\calL_\infty(\nu)}$.
\end{lemma}

\begin{proof}
Since $K$ is bounded and $R$ is square integrable, $(s,t) \mapsto K(s,t)R(s,t)$ defines a square-integrable kernel, from which the compactness of $\mathbf{K}\circ \mathbf{R}$ follows from Lemma~\ref{lem_eigen_com}.
If all eigenvalues of $\mathbf{K}$ are zero, then Lemma~\ref{lem_mercer} implies that for any $t \in T$, $k(t,\cdot) = 0$ $\nu$-a.e., which means that $\mathbf{K}\circ \mathbf{R}$ is the zero operator.
This implies that $0$ is the only eigenvalue of $\mathbf{K}\circ \mathbf{R}$, so the lemma holds true.

Assume that $\mathbf{K}$ has at least one strictly positive eigenvalue, and let $\{\lambda_i\}_{i \in I}$ and $\{e_i\}_{i \in I}$ be given as in Lemma~\ref{lem_mercer_diag}.
Notice that $\|K\|_{\calL_\infty(\nu)} > 0$ holds.
Thus, it suffices to show that $\mu < \|K\|_{\calL_\infty(\nu)}$ holds for any strictly positive eigenvalue $\mu$ of $\mathbf{K}\circ \mathbf{R}$.
Let $\phi \in \calL_2(\nu)$ be the eigenvector of $\mathbf{K}\circ \mathbf{R}$ associated with $\mu$.
Normalizing $\|\phi\|_{\calL_2(\nu)}= 1$, we take the inner product of $\phi$ with each side of the eigenequation $\mu \phi = [\mathbf{K}\circ \mathbf{R}] (\phi)$ to obtain
\begin{align}
\mu &= \left\langle [\mathbf{K}\circ \mathbf{R}] (\phi), \phi \right\rangle_{\calL_2(\nu)} \nonumber \\
&= \int \phi(s) \qty(\int K(s,t)R(s,t)\phi(t) \d \nu(t)) \d \nu(s) \nonumber \\
&= \int \phi(s) \qty(\int \qty(\sum_{i \in I} \lambda_i e_i(s) e_i(t)) R(s,t)\phi(t) \d \nu(t)) \d \nu(s) \nonumber \\
&= \int \int \lim_{n \to \infty} \underbrace{\qty(\sum_{i=1}^n \lambda_i \qty[\phi(s)e_i(s)] R(s,t) \qty[\phi(t)e_i(t)])}_{\eqqcolon \ \ell_n(s,t)} \d \nu(t) \d \nu(s), \label{mu_exp0}
\end{align}
where the third equality follows from Lemma~\ref{lem_mercer}.
Moreover, by Lemma~\ref{lem_mercer_diag}, it follows that
\begin{align*}
\qty|\ell_n(s,t)| &\le \qty|\phi(s)| \qty|R(s,t)|\qty|\phi(t)| \cdot \qty(\sum_{i=1}^n \lambda_i \qty|e_i(s)| \qty|e_i(t)|) \\
&\le \qty|\phi(s)| \qty|R(s,t)|\qty|\phi(t)| \cdot \sqrt{K(s,s)} \sqrt{K(t,t)}
\le \|K\|_{\calL_\infty(\nu)} \cdot \qty|\phi(s)| \qty|R(s,t)|\qty|\phi(t)|,
\end{align*}
where the last inequality holds for every $s,t \in T \setminus N$ with some $\nu$-null set $N \subseteq T$.
Since $\phi$ and $R$ are square integrable, this shows that $|\ell_n|$ are uniformly bounded $\nu \otimes \nu$-a.e.\ by an integrable function.
Hence, the dominated convergence theorem \citep[Theorem 11.21 of][]{ab2006} justifies exchanging the summation and integrations in \eqref{mu_exp0}, from which we have
\begin{align} \label{mu_exp1}
\mu = \sum_{i \in I} \lambda_i \int \int \qty[\phi(s)e_i(s)] R(s,t) \qty[\phi(t)e_i(t)] \d \nu(t) \d \nu(s).
\end{align}

Now, for each $i \in I$, notice that the equivalence class of $\phi e_i$ belongs to $\calL_2(\nu)$, as $e_i$ is bounded $\nu$-a.e.\ due to Lemma~\ref{lem_mercer_diag}.
Hence, (R1) implies that
\begin{align} \label{mu_exp2}
\int \int \qty[\phi(s)e_i(s)] R(s,t) \qty[\phi(t)e_i(t)] \d \nu(t) \d \nu(s)
\le \int \phi(t)^2 e_i(t)^2 \d \nu(t), \quad \forall i \in I,
\end{align}
where the inequality is strict unless $\|\phi e_i\|_{\calL_2(\nu)} = 0$.
Indeed, there must exist at least one index $i \in I$ such that $\lambda_i > 0$ and $\|\phi e_i\|_{\calL_2(\nu)} > 0$, since otherwise, \eqref{mu_exp1} dictates that $\mu = 0$, a contradiction to our initial assumption.
Therefore, by \eqref{mu_exp1} and \eqref{mu_exp2}, again by using the dominated convergence theorem, it follows that
\begin{align*}
\mu %< \sum_{i \in I} \lambda_i \int \phi(t)^2 e_i(t)^2 \d \nu(t) 
< \int \phi(t)^2 \qty(\sum_{i \in I} \lambda_i e_i(t)^2) \d \nu(t) 
\le \int \phi(t)^2 K(t,t) \d \nu(t)
%\le \|K\|_{\calL_\infty(\nu)} \cdot \|\phi\|_{\calL_2(\nu)}^2
\le \|K\|_{\calL_\infty(\nu)}.
\end{align*}
where the first weak inequality follows from Lemma~\ref{lem_mercer_diag}, and the second from the boundedness of $K$ and the normalization of $\phi$.
\end{proof}

%%%%%%%%%%%%%%%%%%%%%%%%%%%%%%%%%%%%%%%%%%%%%
%%%%%%%%%%%%%%%%%%%%%%%%%%%%%%%%%%%%%%%%%%%%%

\section{Proof for Section \ref{sec_pettis}} \label{app_pettis}

\subsection*{Proof of Proposition \ref{prop_pettis}}

As mentioned in the main text, it should be clear that weak measurability implies (P1).
For the converse direction, let $f$ be any process satisfying (P1).
Let $F \coloneqq \overline{\Span}\qty{f(t): t \in T}$ be the closure of the finite linear span of $\{f(t)\}_{t \in T}$, and let $F^\bot$ be the orthogonal complement of $F$.
By Theorem 3 in p.\ 55 of \cite{lax}, for any $x \in X$, there exists a unique orthogonal decomposition $(x',x'') \in F \times F^\bot$ such that $x=x'+x''$.
Then, since $\langle x,f(t) \rangle = \langle x',f(t) \rangle$ holds for every $t \in T$, it is without loss of generality to assume that $x \in F$ to verify that $f$ is weakly measurable.

By the construction of $F$, we can find a sequence $\{x_n\}_{n \in \N}$, converging to $x$ in norm, such that each $x_n$ is expressed as a finite linear combination
\begin{align} \label{y_conv}
x_n = \sum_{k=1}^{k_n} \beta_{n,k} f(t_{n,k})
\end{align}
with some $k_n \in \N$, $\beta_{n,1},\ldots,\beta_{n,k_n} \in \R$, and $t_{n,1},\ldots,t_{n,k_n} \in T$.
Set
\begin{align*}
\bar{x}_n = \frac{1}{n} \sum_{m=1}^n x_m, \quad \forall n \in \N.
\end{align*}
By the triangle inequality, it follows that
\begin{align*}
\|x-\bar{x}_n\| \le \frac{1}{n} \sum_{m=1}^n \|x-x_m\|,
\end{align*}
where the right-hand side is the Ces\`aro mean of $\|x-x_m\|$ through $m=1,\ldots,n$.
Hence, since $x_m \ton x$, it follows that $\bar{x}_n \ton x$.
Then, for each $n$, let $\bar{\phi}_n:T \to \R$ be a function given by
\begin{align*}
\bar{\phi}_n (t) = \frac{1}{n} \sum_{m=1}^n \phi_m(t) \quad \text{where} \quad 
\phi_m (t) = \sum_{k=1}^{k_m} \beta_{m,k} \left\langle f(t_{m,k}), f(t) \right\rangle, \quad \forall t \in T.
\end{align*}
Notice that $\phi_m(t) = \langle x_m, f(t) \rangle$ holds by \eqref{y_conv}, while we have $\lim_{m \to \infty} \langle x_m, f(t) \rangle = \langle x,f(t) \rangle$ since norm convergence implies weak convergence.
Thus, it follows that $\lim_{m \to \infty} \phi_m(t) = \langle x,f(t) \rangle$ for every $t \in T$.
Hence, by the fact that $\bar{\phi}_n(t)$ is the Ces\`aro mean of $\phi_m(t)$, we see that
\begin{align*}
\lim_{n \to \infty} \bar{\phi}_n(t) = \left\langle x,f(t) \right\rangle, \quad \forall t \in T.
\end{align*}
For every $n \in \N$, notice that $\bar{\phi}_n$ is $\nu$-measurable, as being a finite linear combination of functions $t \mapsto \langle f(t_{m,k}), f(t) \rangle$, each of which is $\nu$-measurable by (P1).
Thus, since the mapping $t \mapsto \langle x,f(t) \rangle$ obtains as the pointwise limit of $\bar{\phi}_n$, it is $\nu$-measurable by Lemma 4.29 of \cite{ab2006}.

Next, we show that $f$ is Pettis integrable if it satisfies both (P1) and (P2).
By the Cauchy--Schwartz inequality, we have
\begin{align*}
\int \qty|\langle x,f(t) \rangle| \d\nu(t) \le \|x\| \int \|f(t)\| \d\nu(t),
\end{align*}
where the right-hand side is finite by (P2).
Now, take any sequence $\{x_n\}_{n \in \N} \subseteq X$ such that $x_n \tow y$.
Let us show that $\langle x_n, f(\cdot) \rangle$ converges to $\langle x, f(\cdot)\rangle$ in $L_1$, i.e.,
\begin{align} \label{conv_l1}
\lim_{n \to \infty} \int \qty| \langle x_n-y, f(t) \rangle| \d\nu(t) = 0.
\end{align}
Since any weakly convergent sequence is norm-bounded, we can find some $c \in \R_{++}$ such that $c > \|x\| \lor \sup_{n \in \N} \|y_n\|$.
For each $t \in T$, by the Cauchy--Schwartz inequality, we have
\begin{align*}
\qty| \langle x_n-x, f(t) \rangle| \le \|x_n-y\| \cdot \|f(t)\| < 2c \|f(t)\|,
\end{align*}
whereas $t \mapsto 2c\|f(t)\|$ is integrable by (P2).
Hence, by the dominated convergence theorem, it follows that
\begin{align*}
\lim_{n \to \infty} \int \qty| \langle x_n-x, f(t) \rangle| \d\nu(t)
=  \int \lim_{n \to \infty} \qty| \langle x_n-x, f(t) \rangle| \d\nu(t)
= 0,
\end{align*}
where the second equality holds by $x_n \tow x$.
Thus, (\ref{conv_l1}) is verified.
Hence, Proposition 1 of \cite{huff} implies that $f$ is Pettis integrable.
\hfill {\it Q.E.D.}

\begin{example} \label{ex_nonsep}
We show that (P2) can be violated by Pettis integrable processes.
Let $T=[0,1]$ with the Lebesgue measure, and let $X$ be any non-separable Hilbert space that has an uncountable orthonormal system $\{e(t)\}_{t \in T}$.
Then, let $f(t) = \delta(t)e(t)$, where $\delta: T \to \R$ is any function that need be neither measurable nor integrable.
Regardless of $\delta$, however, $f$ is weakly measurable since $\langle f(s),f(t) \rangle = 0$ whenever $s \neq t$.
Let us show that $f$ is Pettis integrated to zero.
Given any element $x \in X$, consider its orthogonal decomposition $x = x' + x'' \in E \times E^\bot$, where $E \coloneqq \overline{\Span} \{e(t)\}_{t \in T}$.
Then, $\langle x,f(t) \rangle = \langle x',f(t) \rangle$ holds for every $t \in T$.
Moreover, by the construction of $E$, there exists a sequence $\{x_n\}_{n \in \N}$, converging to $x'$ in norm, such that each $x_n$ is a finite linear combination of elements of $\{e(t)\}_{t \in T}$.
Now, for each $n$, we have $\langle x_n, f(t) \rangle = 0$ except for at most finitely many $t$, whence $\int \langle x_n, f(t) \rangle \d\nu(t) = 0$.
Thus, the mapping $t \mapsto |\langle x_n, f(t) \rangle|$ is bounded $\nu$-a.e.\ uniformly across all $n$, from which the dominated convergence theorem implies that
\begin{align*}
0 = \lim_{n\to \infty} \int \langle x_n, f(t) \rangle \d\nu(t)
= \int \langle x', f(t) \rangle \d\nu(t)
= \int \langle x, f(t) \rangle \d\nu(t).
\end{align*}
Since $x$ is arbitrary, we have shown that $\text{w-}\int f(t) \d\nu(t) = 0$.
\hfill $\triangle$
\end{example}

%%%%%%%%%%%%%%%%%%%%%%%%%%%%%%%%%%%%%%%%%%%%%
%%%%%%%%%%%%%%%%%%%%%%%%%%%%%%%%%%%%%%%%%%%%%

\subsection*{Proof of Proposition \ref{prop_var}}
By using the definition of Pettis integral, we obtain \eqref{fubini_cov} as follows:
\begin{align*}
\Cov \qty[x, \int f(t) \d \nu(t)]
&= \E \qty[x \int f(t) \d \nu(t)] - \E \qty[x] \E \qty[\int f(t) \d \nu(t)] \\
&= \int \E \qty[x f(t)] \d \nu(t) - \E \qty[x] \int \E \qty[f(t)] \d \nu(t) \\
&= \int \qty(\E \qty[x f(t)] - \E \qty[x] \E \qty[f(t)]) \d \nu(t) \\
&= \int \Cov \qty[x, f(t)] \d \nu(t),
\end{align*}
where the first and last equalities are due to the definition of covariance.
Moreover, assuming that $s \mapsto \int \Cov \qty[f(s),f(t)] \d \nu(t)$ defines a $\nu$-integrable function, \eqref{fubini_var} is obtained by taking $\int f(s) \d \nu(s)$ in the position of $x$ above and by applying \eqref{fubini_cov}.
\hfill {\it Q.E.D.}

%%%%%%%%%%%%%%%%%%%%%%%%%%%%%%%%%%%%%%%%%%%%%
%%%%%%%%%%%%%%%%%%%%%%%%%%%%%%%%%%%%%%%%%%%%%

\subsection*{Proof of Proposition \ref{prop_cond}}

Suppose that $\hat{f}(\cdot) = \E[f(\cdot) \mid \hat{\Pi}]$ is Pettis integrable.
To verify \eqref{fubini_cond}, fix any event $E \in \hat{\Pi}$.
By the definition of conditional expectation \citep[see, e.g., Chapter 10 of][]{dudley}, it suffices to show that
\begin{align} \label{eq_fubini_E}
\E \qty[\1_E \int f(t) \d\nu(t)]
= \E \qty[\1_E \int \E \qty[f(t) \mid \hat{\Pi}]\d\nu(t)].
\end{align}
Since $\E|\1_E|^2 = \P(E) \le 1$ is finite, the left-hand side is written as
\begin{align*}
\E \qty[\1_E \int f(t) \d\nu(t)] = \int \E \qty[\1_E f(t)] \d\nu(t).
\end{align*}
On the other hand, since $\int \E [f(t) \mid \hat{\Pi}]\d\nu(t)$ is the Pettis integral of the process $\hat{f}$, the right-hand side is written as
\begin{align*}
\E \qty[\1_E \int \E \qty[f(t) \mid \hat{\Pi}]\d\nu(t)]
&= \int \E \qty[\1_E \E \qty[f(t) \mid \hat{\Pi}]] \d\nu(t) \nonumber \\
&= \int \E \qty[ \E \qty[\1_E f(t) \mid \hat{\Pi}]] \d\nu(t)
= \int \E \qty[\1_E f(t)] \d\nu(t),
\end{align*}
where the second equality holds true since $\1_E$ is $f(t)$-measurable, and the third equality comes from the law of iterated expectations.
Hence, (\ref{eq_fubini_E}) is confirmed.
\hfill {\it Q.E.D.}

%%%%%%%%%%%%%%%%%%%%%%%%%%%%%%%%%%%%%%%%%%%%%
%%%%%%%%%%%%%%%%%%%%%%%%%%%%%%%%%%%%%%%%%%%%%

\section{Proof for Section \ref{sec_eq}} \label{app_eq}

\subsection*{Proof of Lemma \ref{lem_r}}
The first assertion directly follows from Lemma \ref{lem_num}.
Next, if $\sup_{s,t \in T} |R(s,t)| < 1$, Lemma \ref{lem_eigen_com} implies that $\|\mathbf{R}\|_{\rm op} < 1$.
Then, again by Lemma \ref{lem_num}, $\|\mathbf{R}\|_{\rm op} < 1$ implies that $\Phi(\mathbf{R}) \subseteq (-1,1)$, from which we conclude that (R1) is satisfied.
Alternatively, if $R$ is multiplicatively separable as in the proposition, then for any $\phi \in \calL_2(\nu) \setminus \{0\}$, we have
\begin{align*}
\left\langle \phi, \mathbf{R}\phi \right\rangle_{\calL_2(\nu)}
= \int \int \phi(s) \qty[rq(s)q(t)] \phi(t) \d \nu(s) \d \nu(t)
= r \cdot \qty(\int q(t)\phi(t) \d \nu(t))^2,
\end{align*}
which has the same sign as $r$.
Given that $r < 0$, (R1) is satisfied.
\hfill {\it Q.E.D.}

\begin{example} \label{ex_uni}
To illustrate the gap between (R1) and (R2), we consider a \emph{uni-directional} payoff structure such that agents are aligned on the line segment, wherein each agent receives strategic influences from only those standing on the right to her.
Specifically, let $T = [0,1]$, $r \in \R$, and $R:[0,1]^2 \to \R$ be given by
\begin{align*}
R(s,t) = r \cdot \1_{\{s < t\}}, \quad \forall s,t \in [0,1].
\end{align*}
For any non-zero constant function for $\phi$, its Rayleigh quotient can be computed as $\langle \phi, \mathbf{R}\phi \rangle_{\calL_2(\nu)} = r/3$, from which (R1) is violated if $r \ge 3$.
On the other hand, (R2) is satisfied with any value of $r \in \R$ since any undirected payoff structure makes the associated integral operator the so-called ``Volterra type,'' which has no non-zero eigenvalue; see Chapter 20.3 of \cite{lax} for details.
\hfill $\triangle$
\end{example}

%%%%%%%%%%%%%%%%%%%%%%%%%%%%%%%%%%%%%%%%%%%%%
%%%%%%%%%%%%%%%%%%%%%%%%%%%%%%%%%%%%%%%%%%%%%

\subsection*{Proof of Proposition \ref{prop_moment}}

Let $f$ be an equilibrium under some payoff structure $R$ and an information environment.
Taking unconditional expectations for both sides of \eqref{br}, the law of iterated expectations and the definition of Pettis integral readily imply \eqref{moment1}.
In particular, \eqref{moment1} indicates that $\E \qty[f(\cdot)]$ solves the functional equation $\phi = \mathbf{R} \phi + \psi$ for $\phi \in \calL_2(\nu)$, where $\psi = \E \qty[\theta(\cdot)]$ belongs to $\calL_2(\nu)$ since $\theta$ satisfies (Q1) and (Q2) by assumption.
Since $R$ is square integrable, by Lemma \ref{lem_eigen_com}, $\mathbf{R}: \calL_2(\nu) \to \calL_2(\nu)$ is a compact linear operator.
Hence, by Theorem 3.4 of \cite{krees}, the equation admits a unique solution if and only if $1$ is not an eigenvalue of $\mathbf{R}$, in which case $\E \qty[f(\cdot)]$ is uniquely determined in $\calL_2(\nu)$.
This means that $\E \qty[f(t)]$ is uniquely determined for $\nu$-almost every $t \in T$, but since \eqref{br} and \eqref{moment1} hold for every agent in an equilibrium, $\E \qty[f(t)]$ is uniquely determined for every $t \in T$.

To obtain the second moment restriction, fix any agent $t \in T$.
By subtracting \eqref{moment1} from \eqref{br}, it follows that
\begin{align*}
f(t) - \E \qty[f(t)] = \E_t \qty[\theta(t) - \E \qty[\theta(t)]] + \E_t \qty[\int R(t,t') \qty(f(t') - \E \qty[f(t')]) \d \nu(t')].
\end{align*}
Noticing that $(f(t)-\E \qty[f(t)])$ is measurable with respect to $t$'s signal, by multiplying both sides of the above equation by $(f(t)-\E \qty[f(t)])$, it follows that
\begin{align*}
\qty|f(t) - \E \qty[f(t)]|^2 = &\E_t \qty[\qty(f(t) - \E \qty[f(t)])\qty(\theta(t) - \E \qty[\theta(t)])] \\
&+ \E_t \qty[\qty(f(t) - \E \qty[f(t)])\int R(t,t') \qty(f(t') - \E \qty[f(t')]) \d \nu(t')],
\end{align*}
and taking unconditional expectations for both sides, the definition of (co)variance yields
\begin{align*}
\Var \qty[f(t)] = \Cov \qty[f(t), \theta(t)]
&+ \E \qty[\qty(f(t) - \E \qty[f(t)])\int R(t,t') \qty(f(t') - \E \qty[f(t')]) \d \nu(t')].
\end{align*}
Moreover, by the definition of Pettis integral, we have
\begin{align*}
&\E \qty[\qty(f(t) - \E \qty[f(t)]) \int R(t,t') \qty(f(t') - \E \qty[f(t')]) \d \nu(t')] \\
&= \int R(t,t') \E \qty[\qty(f(t) - \E \qty[f(t)])\qty(f(t') - \E \qty[f(t')])] \d \nu(t') \\
&= \int R(t,t') \Cov \qty[f(t), f(t')],
\end{align*}
from which \eqref{moment2} is confirmed.
\hfill {\it Q.E.D.}

%%%%%%%%%%%%%%%%%%%%%%%%%%%%%%%%%%%%%%%%%%%%%
%%%%%%%%%%%%%%%%%%%%%%%%%%%%%%%%%%%%%%%%%%%%%

\subsection*{Proof of Proposition \ref{prop_moment2}}

Consider any $(\xi,\zeta)$ as an equilibrium moment under $R$.
Assuming that (R2) is satisfied, Proposition \ref{prop_moment} implies there exists a unique function $\phi:T \to \R$ such that $\phi(t) = \mathbf{R}\phi(t) + \E \qty[\theta]$ holds for all $t \in T$.

Suppose that $\Var\qty[\theta] > 0$.
Then, by computing the Schur complement, the positivity condition implies that for any $n \in \N$ and $t_1,\ldots,t_n \in T$,
\begin{align*}
\mqty[\xi(t_1,t_1) & \ldots & \xi(t_1,t_n) \\ \vdots & \ddots & \vdots \\ \xi(t_n,t_1) & \cdots & \xi(t_n, t_n)] - \frac{1}{\Var\qty[\theta]} \mqty[\zeta(t_1) \\ \vdots \\ \zeta(t_n)] \mqty[\zeta(t_1) \\ \vdots \\ \zeta(t_n)]^\top \succeq O,
\end{align*}
meaning that $\xi$ is a positive semidefinite kernel.
Therefore, by Theorem 12.1.3 of \cite{dudley}, there exists a Gaussian process $\epsilon:T \to X$, which is distributed independently from $\theta$, such that
\begin{align*}
\E \qty[\epsilon(t)] = 0 \quad \text{and} \quad \Cov \qty[\epsilon(s), \epsilon(t)] = \xi(s,t) - \frac{\zeta(s)\zeta(t)}{\Var\qty[\theta]}, \quad \forall s,t \in T.
\end{align*}
Let $f:T \to X$ be a process defined by
\begin{align*}
f(t) = \phi(t) + \frac{\zeta(t)}{\Var\qty[\theta]} \cdot (\theta - \E \qty[\theta]) + \epsilon (t).
\end{align*}
As being the linear transformation of Gaussian random variables, observe that the collection $\{\theta\} \cup \{f(t)\}_{t \in T}$ constitutes a Gaussian process.
By direct calculations, we can confirm that $\E \qty[f(t)] = \phi(t)$, $\Cov \qty[f(s),f(t)] = \xi(s,t)$, and $\Cov \qty[f(t),\theta] = \zeta(t)$ for all $s,t \in T$, so $f$ possesses the desired second moments as in the proposition.
Moreover, since $f(t)$ and $\theta$ are jointly normally distributed, the conditional Gaussian formula yields
\begin{align} \label{cond_f1}
\E \qty[\theta \mid f(t)] = \begin{cases}
\E \qty[\theta] + \frac{\zeta(t)}{\xi(t,t)} \cdot \qty(f(t) - \phi(t)) &\text{if} \quad \xi(t,t) > 0, \\
\E \qty[\theta] &\text{if} \quad \xi(t,t) = 0.
\end{cases}
\end{align}
Similarly, for any $s,t \in T$, we have
\begin{align} \label{cond_f2}
\E \qty[f(s) \mid f(t)] = \begin{cases}
\phi(s) + \frac{\xi(t,s)}{\xi(t,t)} \cdot \qty(f(t) - \phi(t)) &\text{if} \quad \xi(t,t) > 0, \\
\phi(s) &\text{if} \quad \xi(t,t) = 0.
\end{cases}
\end{align}

Let us show that $f$ constitutes an equilibrium in the game $\langle \theta, \{f(t)\}_{t \in T}, R\rangle$, where each agent $t$'s private signal is exactly given as $f(t)$.
To this end, by Corollary \ref{cor_fubini}, we can rewrite the best-response formula \eqref{br} as
\begin{align} \label{br_f}
f(t) = \int R(t,s) \E \qty[ f(s) \mid f(t)] \d \nu(s) + \E \qty[\theta \mid f(t)], \quad \forall t \in T.
\end{align}
If $\xi(t,t) = 0$, then $\Var\qty[f(t)] = 0$ so that $f(t) = \phi(t)$ holds $\P$-a.s.
Moreover, from \eqref{cond_f1} and \eqref{cond_f2}, the right-hand side of \eqref{br_f} is written as $\mathbf{R}\phi(t) + \E \qty[\theta]$, which coincides with $\phi(t)$ by the construction of $\phi$.
Otherwise, if $\xi(t,t) \neq 0$, again by \eqref{cond_f1} and \eqref{cond_f2}, the right-hand side of \eqref{br_f} can be arranged as
\begin{align*}
&\int R(t,s) \qty(\phi(s) + \frac{\xi(s,t)}{\xi(t,t)} \cdot \qty(f(t) - \phi(t))) \d \nu(s) + \E \qty[\theta] + \frac{\zeta(t)}{\xi(t,t)} \cdot \qty(f(t) - \phi(t)) \\
&\quad = \underbrace{\qty(\int R(t,s)\phi(s) \d \nu(s) + \E \qty[\theta])}_{(\dagger)} + \frac{f(t) - \phi(t)}{\xi(t,t)} \cdot \underbrace{\qty(\int R(t,s)\xi(t,s) \d \nu(s) + \zeta(t))}_{(\ddagger)},
\end{align*}
where $(\dagger)$ is equal to $\phi(t)$ by the construction of $\phi$, and $(\ddagger)$ is equal to $\xi(t,t)$ by the obedience condition.
Therefore, the above expression coincides with $f(t)$, from which we conclude that \eqref{br_f} holds true.

Lastly, let us briefly discuss the case of $\Var \qty[\theta] = 0$.
In this case, the positivity condition dictates that $\zeta(t) = 0$ for all $t \in T$.
Then, by letting $\epsilon: T \to X$ be a unbiased Gaussian process, having $\xi$ as its covariance function, we can confirm that $f(t) = \phi(t) + \epsilon(t)$ works as the desired signal structure and equilibrium.
\hfill {\it Q.E.D.}

%%%%%%%%%%%%%%%%%%%%%%%%%%%%%%%%%%%%%%%%%%%%%
%%%%%%%%%%%%%%%%%%%%%%%%%%%%%%%%%%%%%%%%%%%%%

\subsection*{Proof of Proposition \ref{prop_unique}}

\subsubsection*{Sufficiency}
Given any information environment and $R$ satisfying (R1), suppose that $f$ and $g$ are two equilibria such that $(s,t) \mapsto \E \qty[f(s)g(t)]$ is jointly measurable.
Let $h=f-g$, and let $\sigma:T \to \R$ and $\rho:T^2 \to \R$ be defined by
\begin{align*}
\sigma(t) = \sqrt{\E\qty|h(t)|^2} \quad \text{and} \quad 
\rho(s,t) = \begin{cases}
\frac{\E \qty[h(s)h(t)]}{\sigma(s) \sigma(t)} &\text{if} \quad \sigma(s)\sigma(t) > 0, \\
0 &\text{otherwise}.
\end{cases}
\end{align*}
Since $f$ and $g$ satisfy (Q1), and since $(s,t) \mapsto \E \qty[f(s)g(t)]$ is jointly measurable, $\sigma$ is $\nu$-measurable and $\rho$ is jointly measurable.
Moreover, since the mean of $h$ is zero by Proposition \ref{prop_moment}, notice that $\qty|\sigma(t)|^2$ expresses the variance of $h(t)$, which is bounded by the sum of $\E \qty|f(t)|^2$ and $\E \qty|g(t)|^2$ due to the Cauchy--Schwartz inequality.
Then, since $f$ and $g$ satisfy (Q2), it follows that $\sigma$ belongs to $\calL_2(\nu)$.
Additionally, Proposition \ref{prop_moment} implies that $\rho(s,t)$ is the correlation coefficient between $h(s)$ and $h(t)$, where the convention of $\rho(s,t) = 0$ applies if either $h(s)$ or $h(t)$ has zero variance.
Hence, we see that $\rho$ is positive semidefinite and bounded on the diagonal with $\|\rho\|_{\calL_{\infty}(\nu)} \le 1$.

Now, since both $f$ and $g$ satisfy the best-response formula \eqref{br}, subtracting one from the other, the linearity of Pettis integral implies that
\begin{align*}
f(t) - g(t) = \E_t \qty[\int R(t,t') \qty(f(t') - g(t')) \d \nu(t')], \quad \forall t \in T.
\end{align*}
Noticing that $f(t)$ and $g(t)$ are $x(t)$-measurable, by multiplying both sides of the equation by $(f(t)-g(t))$, we have
\begin{align*}
\qty|f(t)-g(t)|^2 = \E_t \qty[\qty(f(t)-g(t)) \int R(t,t') \qty(f(t')-g(t')) \d \nu(t')], \quad \forall t \in T.
\end{align*}
Then, by taking unconditional expectations for both sides, it follows that
\begin{align*}
\E \qty|f(t)-g(t)|^2 = \int R(t,t') \E \qty[(f(t)-g(t))\qty(f(t')-g(t'))] \d \nu(t'), \quad \forall t \in T.
\end{align*}
Therefore, by the constructions of $\sigma$ and $\rho$, we obtain
\begin{align} \label{eq2}
\sigma(t) = \int \qty[R(t,t')\rho(t,t')] \sigma(t') \d \nu(t'), \quad \forall t \in T.
\end{align}

To conclude the proof of sufficiency, let us show that $\sigma(t) = 0$ for all $t \in T$.
To this end, we identify $\rho$ as the integral operator, acting on $\calL_2(\nu)$, and consider the following two possible cases.
\begin{itemize}
\item {\it Case 1: $\rho$ has no non-zero eigenvalue.}
In this case, Lemma~\ref{lem_mercer} implies that for any $t \in T$, $\rho(t,\cdot) = 0$ $\nu$-a.e.
By this and \eqref{eq2}, we readily see that $\sigma(t) = 0$ for all $t \in T$.
\item {\it Case 2: $\rho$ has at least one non-zero eigenvalue.}
Let $\mathbf{Q}: \calL_2(\nu) \to \calL_2(\nu)$ be a linear operator given by
\begin{align*}
\mathbf{Q}\phi = \int \qty[R(\cdot,t)\rho(\cdot,t)] \phi(t) \d \nu(t), \quad \forall \phi \in \calL_2(\nu).
\end{align*}
Since $R$ satisfies (R1), Lemma~\ref{lem_eigen} implies that all real eigenvalues of $\mathbf{Q}$ is strictly less than $\|\rho\|_{\calL_{\infty}(\nu)}$.
In particular, since $\|\rho\|_{\calL_{\infty}(\nu)} \le 1$, this means that $1$ cannot be an eigenvalue of  $\mathbf{Q}$.
Therefore, \eqref{eq2} implies that $\|\sigma\|_{\calL_2(\nu)} = 0$
This means that $\sigma(t) = 0$ holds for $\nu$-almost every $t$, but since \eqref{eq2} must be satisfied for every $t$, we have $\sigma(t) = 0$ for every $t$.
\end{itemize}

\subsubsection*{Necessity}
Suppose that (R2) is violated so that $\mathbf{R}$ has some eigenvalue $\lambda \ge 1$.
Consider any eigenvector $\phi \in \calL_2(\nu)$ with $\|\phi\|_{\calL_2(\nu)} > 0$ such that
\begin{align} \label{eigen_phi}
\lambda \phi(t) = \int R(t,t') \phi(t') \d \nu(t'), \quad \forall t \in T.
\end{align}
Define $K:T^2 \to \R$ by $K(t,t)=1$ and $K(t,t')=1/\lambda$ for all $t,t' \in T$ with $t \neq t'$.
Noticing that $1/\lambda \in (0,1]$, we can easily verify that $K$ is a positive semidefinite kernel.
Hence, by Theorem 12.1.3 of \cite{dudley}, there exists a probability space $(\Omega,\Pi,\P)$ on which a Gaussian process $\{x(t)\}_{t \in T}$, having mean $0$ and covariance $K$, takes place.
We let $x(t)$ be the agent $t$'s private signal.
%In addition, let the state $\theta(t)$ be common across agents and given as a deterministic variable $\bar{\theta} \in \R$.

Now, suppose that there exists at least one equilibrium $f$.
We then consider an alternative strategy profile $g$ such that
\begin{align*}
g(t) = f(t) + \phi(t) x(t), \quad \forall t \in T.
\end{align*}
Since $f$ is a regular strategy profile, so must be $g$.
Moreover, by the linearity of Pettis integral, the local aggregate of each agent is given as
\begin{align*}
G(t) \coloneqq \int R(t,t') g(t') \d\nu(t')
= F(t)+ \int R(t,t')\phi(t')x(t') \d\nu(t').
\end{align*}
By adding $\theta(t)$ and taking agent $t$'s conditional expectation, it follows that
\begin{align*}
\E_t \qty[G(t)] + \E_t \qty[\theta(t)]
= \underbrace{\qty(\E_t \qty[F(t)] + \E_t \qty[\theta(t)])}_{(*)} + \E_t \qty[\int R(t,t')\phi(t')x(t') \d\nu(t')],
\end{align*}
while $(*)$ is equal to $f(t)$ since $f$ is an equilibrium.
Moreover, since $\E_t\qty[x(t')] = x(t)/\lambda$ holds by the conditional Gaussian formula, by using Corollary \ref{cor_fubini}, it follows that
\begin{align*}
\E_t \qty[\int R(t,t')\phi(t')x(t') \d\nu(t')] = \frac{x(t)}{\lambda} \int R(t,t')\phi(t')\d\nu(t') = \phi(t) x(t),
\end{align*}
where the last equality holds by \eqref{eigen_phi}.
Hence, we have
\begin{align*}
\E_t \qty[G(t)] + \E_t \qty[\theta(t)] = f(t) + \phi(t)x(t) ,\quad \forall t \in T.
\end{align*}
Since the above right-hand side is equal to $g(t)$ by construction, we conclude that $g$ is an equilibrium.
Notice that $\Var \qty[f(t)-g(t)] = |\phi(t)|^2$ and $\|\phi\|^2_{\calL_2(\nu)} > 0$, which in turn imply that $f$ and $g$ must be different on a set of positive measure.
In particular, our construction reveals that there are a continuum of different equilibria since we can consider an arbitrary eigenvector that is proportional to $\phi$.
\hfill {\it Q.E.D.}

%%%%%%%%%%%%%%%%%%%%%%%%%%%%%%%%%%%%%%%%%%%%%
%%%%%%%%%%%%%%%%%%%%%%%%%%%%%%%%%%%%%%%%%%%%%

\section{Proof for Section \ref{sec_info}} \label{app_info}

\subsection*{Proof of Lemma \ref{lem_tg_eq}}

Given $f$ as in the proposition, let us show that this strategy profile constitutes an equilibrium.
By Proposition~\ref{prop_var}, the aggregate is calculated as 
\begin{align*}
F \coloneqq \int_0^1 f(t) \d \nu(t) = \frac{m \theta}{1-rm}.
\end{align*}
Since $\E \qty[\theta] = 0$ by assumption, we have $\E \qty[F] = 0$, from which the best-response formula \eqref{br} is confirmed for all uninformed agents $t \in [0,1] \setminus M$.
Moreover, since all informed agents $t \in M$ face no uncertainties about $F$ and $\theta$, it follows that
\begin{align*}
r \E_t \qty[\int_0^1 f(t) \d \nu(t)] + \E_t \qty[\theta]
= \frac{rm \theta}{1-rm} + \theta
= \frac{\theta}{1-rm}
=f(t),
\end{align*}
from which the best-response forma \eqref{br} is confirmed for informed agents as well.
Therefore, $f$ constitutes an equilibrium, which is unique by Proposition \ref{prop_unique}.\footnote{More formally, to apply Proposition~\ref{prop_unique}, we need to verify that the measurability condition in the proposition is satisfied by $f$ and any candidate equilibrium $g$. This can be checked as follows: Since $g$ is an equilibrium, Proposition~\ref{prop_moment} implies that $\E \qty[g(t)] = 0$ for all $t \in [0,1]$. Then, by \eqref{br}, it follows that $g(t) = 0$ for all uninformed agents $t \in [0,1] \setminus M$. Consequently, we have
\begin{align*}
\E \qty[f(s)g(t)] = \begin{cases}
\frac{\E \qty[\theta g(t)]}{1-rm} &\text{if} \quad s,t \in M, \\
0 &\text{otherwise}.
\end{cases}
\end{align*}
From this, we confirm that $(s,t) \mapsto \E \qty[f(s)g(t)]$ is a jointly measurable mapping since $g$ satisfies (Q1) as required by Definition~\ref{def_st}.}
The associated equilibrium moment can be easily calculated by noticing that $\E \qty[\theta] = 0$ and $\Var \qty[\theta] = 1$.
\hfill {\it Q.E.D.}

%%%%%%%%%%%%%%%%%%%%%%%%%%%%%%%%%%%%%%%%%%%%%
%%%%%%%%%%%%%%%%%%%%%%%%%%%%%%%%%%%%%%%%%%%%%

\subsection*{Proof of Lemma \ref{lem_tg}}

By taking the first-order derivative,
\begin{align*}
V_{\rm tg}'(m) = \frac{\alpha + (r\alpha - 2 \beta)m}{(1-rm)^3}.
\end{align*}
Observe that the numerator is linear in $m$, while the denominator is positive since $r < 1$ and $m \in [0,1]$.
Hence, $V_{\rm tg}$ admits at most one interior maximizer.

First, suppose that $V_{\rm tg}'(0) > 0$ and $V_{\rm tg}'(1) < 0$.
Since $V_{\rm tg}'(0) = \alpha$, we have $V_{\rm tg}'(0) > 0$ if and only if $\alpha > 0$.
Moreover, observe that
\begin{align*}
V_{\rm tg}'(1) = \frac{\alpha + (r\alpha - 2 \beta)}{(1-r)^3} \propto (1+r) \alpha - 2\beta,
\end{align*}
from which $V_{\rm tg}'(1) < 0$ if and only if $\beta > (\frac{1+r}{2}) \cdot \alpha$.
Hence, (T2) is satisfied if and only if $V_{\rm tg}'(0) > 0$ and $V_{\rm tg}'(1) < 0$, under which $V_{\rm tg}(m)$ is maximized at an interior solution.
Moreover, the interior solution is uniquely pinned down by the first-order condition $V'_{\rm tg} = 0$, which gives rise to $m^* = \frac{\alpha}{2\beta - r \alpha}$.
Notice that $m^* \in (0,1)$ holds under (T2), and that $V_{\rm tg}(m^*) = \frac{\alpha^2}{4(\beta-r\alpha)}$ by direct calculations.

Second, suppose that $V_{\rm tg}'(0) \le 0$ or $V_{\rm tg}'(1) \ge 0$.
In this case, $V_{\rm tg}(m)$ is maximized at a corner solution, and hence,
\begin{align*}
V_{\rm tg}^* = \max \qty{V_{\rm tg}(0), V_{\rm tg}(1)} = \max \qty{0,\, \frac{\alpha - \beta}{(1-r)^2}}.
\end{align*}
Therefore, it follows that full disclosure (resp.\ no disclosure) is optimal if and only if
\begin{align*}
\alpha \underset{(\le)}{\ge} \beta \quad \text{and} \quad \qty[\alpha \le 0 \quad \text{or} \quad \beta \le \qty(\frac{1+r}{2}) \cdot \alpha],
\end{align*}
which can be reduced to (T1) (resp.\ (T3)).
\hfill {\it Q.E.D.}

%%%%%%%%%%%%%%%%%%%%%%%%%%%%%%%%%%%%%%%%%%%%%
%%%%%%%%%%%%%%%%%%%%%%%%%%%%%%%%%%%%%%%%%%%%%

\subsection*{Proof of Proposition \ref{prop_tg}}

We begin the proof with a few auxiliary lemmas that are derived as consequences of obedience and positivity.

\begin{lemma} \label{lem_h}
Given any equilibrium moment $(\xi,\zeta)$, let $\kappa(s,t) = \xi(s,t) - \zeta(s)\zeta(t)$ for all $s,t \in [0,1]$. Then, $\kappa:[0,1]^2 \to \R$ is a positive semidefinite kernel.
\end{lemma}

\begin{proof}
The lemma immediately follows from the positivity of $(\xi,\zeta)$ and the property of Schur complements.
\end{proof}

\begin{lemma} \label{lem_bound}
For any equilibrium moment $(\xi,\zeta)$, it holds that
\begin{align*}
\qty(\int_0^1 \zeta(t) \d t)^2 \le \int_0^1 \int_0^1 \xi(s,t) \d s \d t \le
\min \qty{\int_0^1 \xi(t,t) \d t,\, \qty(\frac{1}{1-r})^2}.
\end{align*}
\end{lemma}

\begin{proof}
Since $\xi:[0,1]^2 \to \R$ is a positive semidefinite kernel, the Cauchy--Schwartz inequality \eqref{k_cs} implies that
\begin{align} \label{bound1}
\int_0^1 \int_0^1 \xi(s,t) \d s \d t \le \int_0^1 \xi(t,t) \d t.
\end{align}
Let $\bar{T} \coloneqq [0,1] \cup \{\s\}$, where ``$\s$'' can be thought of as the dummy index representing the state, and endow this set with a probability measure $\bar{\nu}$ that results in a uniform distribution over $[0,1]$ and a point mass at $\s$ with equal probabilities.
Given any equilibrium moment $(\xi,\zeta)$, we define a bivariate function $\Xi:\bar{T}^2 \to \R$ by
\begin{align*}
\Xi(s,t) \coloneqq \begin{cases}
\xi(s,t) &\text{if} \quad s,t \in [0,1], \\
\zeta(s) &\text{if} \quad s \in [0,1] \quad \text{and} \quad t = \s, \\
\zeta(t) &\text{if} \quad s = \s \quad \text{and} \quad t \in [0,1], \\
1 &\text{if} \quad s=t=\s.
\end{cases}
\end{align*}
By the positivity of $(\xi,\zeta)$, we see that $\Xi$ is a positive semidefinite kernel.
Hence, for any bounded function $\phi:T \to \R$, the constructions of $\Xi$ and $\bar{\nu}$ yield
\begin{align*}
0 &\le \int_{\bar{T}} \int_{\bar{T}} \phi(s) \Xi(s,t) \phi(t) \d\bar{\nu}(s) \d\bar{\nu}(t) \\
&= \int_0^1 \int_0^1 \phi(s)\xi(s,t)\phi(t) \d i \d j + 2 \phi(\s) \int_0^1 \phi(t) \zeta(t) \d t + \phi(\s)^2.
\end{align*}
In particular, by letting $\phi(t) = 1$ for all $t \in [0,1]$ and $\phi(\s) = - (\int_0^1 \zeta(t) \d t)$, we obtain
\begin{align} \label{bound2}
\int_0^1 \int_0^1 \xi(s,t) \d s \d t \ge - 2 \phi(\s) \int_0^1 \zeta(t) \d t - \phi(\s)^2 =  \qty(\int_0^1 \zeta(t) \d t)^2.
\end{align}
Therefore, our remaining task is to show that 
\begin{align} \label{bound3}
\int_0^1 \int_0^1 \xi(s,t) \d s \d t \le \qty(\frac{1}{1-r})^2.
\end{align}
To this end, we consider two cases separately.

First, suppose that $r \in [0,1)$.
By integrating the obedience constraint, it follows that
\begin{align*}
\int_0^1 \xi(t,t) \d t &= r \int_0^1\int_0^1 \xi(t,t') \d t \d t' + \int_0^1 \zeta(t) \d t \\
&\le r \int_0^1\int_0^1 \xi(t,t') \d t \d t' + \qty(\int_0^1\int_0^1 \xi(t,t') \d t \d t')^{\frac{1}{2}} \\
&\le r \int_0^1 \xi(t,t) \d t + \qty(\int_0^1 \xi(t,t) \d t)^{\frac{1}{2}},
\end{align*}
where the two inequalities can be confirmed by \eqref{bound1}, \eqref{bound2}, and $r \ge 0$.
Thus, by letting $z = (\int_0^1 \xi(t,t) \d t)^{1/2}$, we have
\begin{align*}
(1-r)z^2 - z \le 0.
\end{align*}
Noticing that $1-r > 0$, we can see that this quadratic inequality is satisfied if and only if $0 \le z \le \frac{1}{1-r}$.
Hence, by the definition of $z$, we have $\int_0^1 \xi(t,t) \d t \le (\frac{1}{1-r})^2$, and particularly, by \eqref{bound1}, we observe that \eqref{bound3} holds true.

Second, suppose that $r < 0$.
Then, by \eqref{bound1}, \eqref{bound2}, and obedience, it follows that
\begin{align*}
\qty(\int_0^1 \zeta(t) \d t)^2 &\le \int_0^1 \xi(t,t) \d t \\
&= r \int_0^1 \int_0^1 \xi(s,t) \d s \d t + \int_0^1 \zeta(t) \d t \\
&\le r \qty(\int_0^1 \zeta(t) \d t)^2 + \int_0^1 \zeta(t) \d t,
\end{align*}
where the last inequality follows from \eqref{bound2} and $r < 0$.
Since $1-r > 0$, by the same reasoning as before, but now by letting $z =  \int_0^1 \zeta(t) \d \nu(t)$, it follows that $\int_0^1 \zeta(t) \d t \le \frac{1}{1-r}$.
Together with this, by using obedience and \eqref{bound1}, it follows that
\begin{align*}
\frac{1}{1-r} \ge \int_0^1 \zeta(t) \d t = \int_0^1 \xi(t,t) \d t -r \int_0^1 \int_0^1 \xi(s,t) \d s \d t
\ge (1-r) \int_0^1 \int_0^1 \xi(s,t) \d s \d t,
\end{align*}
from which \eqref{bound3} is confirmed.
\end{proof}

In what follows, let $(\xi,\zeta)$ be an arbitrarily fixed equilibrium moment.
We show that $V(\xi,\zeta)$ is bounded by $V_{\rm tg}^*$, which is given as in Lemma~\ref{lem_tg}, by discussing each of three possible cases one by one.

\subsubsection*{No Disclosure}

Assuming (T1), we want to show that $V(\xi,\zeta) \le 0$.
By substituting the obedience condition into $V(\xi,\zeta)$, it follows that
\begin{align}
V(\xi,\zeta) &= u \int_0^1 \int_0^1 \xi(s,t) \d s \d t + v \qty(r \int_0^1 \int_0^1 \xi(s,t) \d s \d t + \int_0^1 \zeta(t) \d t )+ w \int_0^1 \zeta(t) \d t \nonumber \\
&= (v+w) \int_0^1 \zeta(t) \d t + (r(v+w)-(rw-u)) \int_0^1 \int_0^1 \xi(s,t) \d s \d t \nonumber \\
&= \alpha \int_0^1 \zeta(t) \d t + (r \alpha - \beta) \int_0^1 \int_0^1 \xi(s,t) \d s \d t, \label{val1}
\end{align}
where the last line is due to the definitions of $\alpha$ and $\beta$.
Moreover, slightly arranging this expression, we get
\begin{align}
V(\xi,\zeta)
&= (\alpha-\beta) \underbrace{\int_0^1 \int_0^1 \xi(s,t) \d s \d t}_{\text{(I-a)}} + \alpha \underbrace{\int_0^1 \qty(\zeta(s) - (1-r) \int_0^1 \xi(s,t) \d t) \d s}_{\text{(I-b)}}. \label{val2}
\end{align}
By (T1), we know that $(\alpha - \beta) \le 0$ and $\alpha \le 0$.
In addition, (I-a) is non-negative since $\xi$ is positive semidefinite.
Furthermore, by obedience, we can express (I-b) as
\begin{align*}
\text{(I-b)} = \int_0^1 \xi(t,t) \d t - \int_0^1 \int_0^1 \xi(s,t) \d s \d t,
\end{align*}
which is non-negative by Lemma~\ref{lem_bound}.
Therefore, combining these observations, we conclude that $V(f,g) \le 0$.

\subsubsection*{Partial Disclosure}

Assuming (T2), we want to show that $V(\xi,\zeta) \le \frac{\alpha^2}{4(\beta-r\alpha)}$.
Let $\kappa$ be a positive semidefinite kernel defined as in Lemma~\ref{lem_h}.
Then, we can arrange \eqref{val1} as
\begin{align*}
V(\xi,\zeta) = \alpha \int_0^1 \zeta(t) \d t - (\beta - r\alpha) \qty(\int_0^1 \zeta(t) \d t)^2 - (\beta-r\alpha) \underbrace{\int_0^1 \int_0^1 \kappa(s,t) \d s \d t}_{\rm (II)},
\end{align*}
Since $\frac{1+r}{2} > r$ by $r < 1$, and since $\alpha > 0$ by (T2), we have $(\frac{1+r}{2}) \cdot \alpha > r \alpha$.
Moreover, since $\beta > (\frac{1+r}{2}) \cdot \alpha$ by (T2), it follows that 
$(\beta - r \alpha)$ is strictly positive, whereas (II) is non-negative since $h$ is positive semidefinite.
Hence, by ignoring the last term, we observe that $V(\xi,\zeta)$ is bounded from above by
\begin{align*}
\max_{z \in \R} \quad \alpha z - (\beta - r\alpha) z^2,
\end{align*}
which attains the maximum of $\frac{\alpha^2}{4(\beta-r\alpha)}$ at $z^* = \frac{\alpha}{2(\beta - r \alpha)}$, as desired.

\subsubsection*{Full Disclosure}

Assuming (T3), we want to show that $V(\xi,\zeta) \le \frac{\alpha-\beta}{(1-r)^2}$.
Let us consider two subcases.

First, suppose that $\alpha \ge 0$.
Continuing from the case of no disclosure, by substituting (I-b) into \eqref{val2}, it follows that
\begin{align} \label{val3}
V(\xi,\zeta) = \alpha \int_0^1 \xi(t,t) \d t - \beta \int_0^1 \int_0^1 \xi(s,t) \d s \d t.
\end{align}
Again, by letting $\kappa$ be as in Lemma~\ref{lem_h}, we can arrange \eqref{val3} as
\begin{align*}
V(\xi,\zeta) = (\alpha - \beta) \qty(\int_0^1 \zeta(t) \d t)^2 - \underbrace{\qty(\alpha \int_0^1 \kappa(t,t) \d t - \beta \int_0^1 \int_0^1 \kappa(s,t) \d s \d t)}_{\text{(III-a)}},
\end{align*}
whereas we have
\begin{align*}
\text{(III-a)} = \alpha \qty(\int_0^1 \kappa(t,t) \d t - \int_0^1 \int_0^1 \kappa(s,t) \d s \d t) + (\alpha - \beta) \int_0^1 \int_0^1 \kappa(s,t) \d s \d t.
\end{align*}
Since $\alpha \ge 0$ by assumption, and since $\alpha \ge \beta$ by (T3), we see that (III) is non-negative, as $\kappa$ is a positive semidefinite kernel.
Therefore, by this and Lemma~\ref{lem_bound}, we conclude that
\begin{align*}
V(\xi,\zeta) \le (\alpha - \beta) \qty(\int_0^1 \zeta(t) \d t)^2 \le \frac{\alpha-\beta}{(1-r)^2},
\end{align*}
as desired.

Second, suppose that $\alpha < 0$.
By arranging \eqref{val3}, it holds that
\begin{align*}
V(\xi,\zeta) &= (\alpha - \beta) \int_0^1 \int_0^1 \xi(s,t) \d s \d t + \alpha \underbrace{\qty(\int_0^1 \xi(t,t) \d t - \int_0^1 \int_0^1 \xi(s,t) \d s \d t)}_{\text{(III-b)}},
\end{align*}
where (III-b) is non-negative since $\xi$ is a positive semidefinite kernel.
Moreover, since $\alpha < 0$ by assumption, it follows that
\begin{align*}
V(\xi,\zeta) \le (\alpha-\beta) \int_0^1 \int_0^1 \xi(s,t) \d s \d t 
\le \frac{\alpha-\beta}{(1-r)^2},
\end{align*}
where the last inequality is due to Lemma~\ref{lem_bound}.
\hfill {\it Q.E.D.}

%%%%%%%%%%%%%%%%%%%%%%%%%%%%%%%%%%%%%%%%%%%%%
%%%%%%%%%%%%%%%%%%%%%%%%%%%%%%%%%%%%%%%%%%%%%

\subsection*{Proof of Proposition \ref{prop_sym}}

%\begin{align*}
%\mqty[\bar{\xi}_1 & \bar{\xi}_2 & \cdots & \bar{\xi}_2 & \bar{\zeta} \\
%\bar{\xi}_2 & \bar{\xi}_1 & \cdots & \bar{\xi}_2 & \bar{\zeta} \\
%\vdots & \vdots & \ddots & \vdots & \vdots \\
%\bar{\xi}_2 & \bar{\xi}_2 & \cdots & \bar{\xi}_1 & \bar{\zeta} \\
%\bar{\zeta} & \bar{\zeta} &\cdots & \bar{\zeta} & 1 ] \succeq O
%\end{align*}

By adopting Proposition 1 of \cite{bm2013}, we can show that any $(\bar{\xi}_1,\bar{\xi}_2,\bar{\zeta}) \in \R^3$ satisfies obedience and positivity if and only if
\begin{align*}
\bar{\xi}_1 = r \bar{\xi}_2 + \bar{\zeta} \quad \text{and} \quad \bar{\xi}_1 \ge \bar{\xi}_2 \ge \bar{\zeta}^2.
\end{align*}
Then, given any $m \in [0,1]$, we can readily see that these conditions are satisfied by $(\bar{\xi}_1,\bar{\xi}_2,\bar{\zeta})$ in the proposition, so it is a symmetric equilibrium moment.
Moreover, by substituting $(\bar{\xi}_1,\bar{\xi}_2,\bar{\zeta})$ into the designer's objective function, we observe that
\begin{align*}
V(\bar{\xi}_1,\bar{\xi}_2,\bar{\zeta})
&= v \bar{\xi}_1 + u \bar{\xi}_2 + w \bar{\zeta}
%&= v \cdot \frac{m}{(1-rm)^2} + u \cdot \frac{m^2}{(1-rm)^2} + w \cdot \frac{m}{1-rm}
= \frac{(v+w)m - (u-rw)m^2}{(1-rm)^2},
\end{align*}
which coincides with $V_{\rm tg}(m)$, calculated as in \eqref{prob_tg}.

In the remainder, assume that $\theta \sim \calN(0,1)$.
Let $\{x(t)\}_{t \in T}$ be the signal structure, presented as in the proposition.
Since $\E \qty[x(t)] = 0$, and since $\epsilon(t)$ are independent from $\theta$ and each other, we observe that
\begin{gather}
\Cov \qty[x(t),\theta] = \frac{m}{1-rm} = \bar{\zeta}, \label{moment_sym1} \\
\Cov \qty[x(t),x(t')] = \qty(\frac{m}{1-rm})^2 = \bar{\xi}_2. \label{moment_sym2}
\end{gather}
Moreover, since $\Var \qty[\theta] = \Var\qty[\epsilon(t)] = 1$, we have
\begin{align} \label{moment_sym3}
\Var\qty[x(t)] = \qty(\frac{m}{1-rm})^2 + \qty(\frac{\sqrt{m(1-m)}}{1-rm})^2
= \frac{m}{(1-rm)^2} = \bar{\xi}_1.
\end{align}
In addition, by Proposition~\ref{prop_var}, notice that $\{x(t')\}_{t' \in [0,1]}$ are Pettis integrated to $\int_0^1 x(t') \d \d t' = \frac{m \theta}{1-rm}$.
Hence, by \eqref{moment_sym1} and \eqref{moment_sym2}, the conditional Gaussian formula implies
\begin{align*}
&r \E \qty[\int_0^1 x(t') \d t' \mid x(t)] + \E \qty[\theta \mid x(t)]
= \qty(\frac{rm}{1-rm} + 1) \cdot \E \qty[\theta \mid x(t)] \\
&\qquad = \frac{1}{1-rm} \cdot \frac{\Cov \qty[x(t),\theta]}{\Var\qty[x(t)]} \cdot x(t)
= \frac{1}{1-rm} \cdot \frac{m}{1-rm} \cdot \frac{(1-rm)^2}{m} \cdot x(t)
= x(t),
\end{align*}
which shows that $\{x(t)\}_{t \in [0,1]}$ constitutes an equilibrium under the signal structure $\{x(t)\}_{t \in [0,1]}$.
In particular, from \eqref{moment_sym1}, \eqref{moment_sym2}, and \eqref{moment_sym3}, we confirm that the second moment of this equilibrium is characterized by $(\bar{\xi}_1,\bar{\xi}_2,\bar{\zeta})$.
\hfill {\it Q.E.D.}

%\newpage
\bibliography{reference}

\end{document}